\documentclass[preprint,onecolumn,preprintnumbers,superscriptaddress,nofootinbib]{revtex4} 

\usepackage[utf8]{inputenc}
\usepackage{amsmath, amsthm, amssymb, bm}
\usepackage{mathrsfs}
\usepackage{extarrows}
\usepackage{pifont}
\usepackage{diagbox}
\usepackage{xcolor}
\usepackage{multirow}
\usepackage{graphicx}
\usepackage{hyperref}
\hypersetup{  colorlinks=true, linkcolor=blue, citecolor=red, urlcolor=blue }

\usepackage{braket}
\usepackage{pgfplots}
\input{qcircuit}

\usepackage[titletoc,title]{appendix}
\usepackage[ruled,linesnumbered]{algorithm2e}

\newtheorem{theorem}{Theorem}
\newtheorem{lemma}{Lemma}

\newtheorem{problem}{Problem}
\newtheorem{corollary}{Corollary}
\newtheorem{claim}{Claim}

\newtheorem{assumption}{Assumption}

\DeclareMathOperator*{\argmin}{\arg\,\min}

\DeclareMathOperator{\diag}{diag}

\DeclareMathOperator{\rank}{rank}

\newcommand{\abs}[1]{\left| #1 \right|}
\newcommand{\vabs}[1]{\left\| #1 \right\|}
\newcommand{\pbra}[1]{\left( #1 \right)}
\newcommand{\cbra}[1]{\left\{ #1 \right\}}

\newcommand{\ceil}[1]{\left\lceil #1\right \rceil }

\newcommand{\Cbb}{\mathbb{C}}
\newcommand{\Ebb}{\mathbb{E}}

\newcommand{\Rbb}{\mathbb{R}}

\newcommand{\Ccal}{\mathcal{C}}

\newcommand{\Tcal}{\mathcal{T}}
\newcommand{\Amat}{\bm{\mathrm A}}

\newcommand{\Pmat}{\bm{\mathrm P}}
\newcommand{\Rmat}{\bm{\mathrm R}}

\newcommand{\Mmat}{\bm{\mathrm M}}
\newcommand{\Hmat}{\bm{\mathrm H}}
\newcommand{\Imat}{\bm{\mathrm I}}
\newcommand{\Qmat}{\bm{\mathrm Q}}
\newcommand{\Umat}{\bm{\mathrm U}}
\newcommand{\Vmat}{\bm{\mathrm V}}
\newcommand{\Wmat}{\bm{\mathrm W}}
\newcommand{\Xmat}{\bm{\mathrm X}}

\newcommand{\poly}{\mathrm{poly}}
\newcommand{\polylog}{\mathrm{polylog}}
\newcommand{\real}{\mathrm{Re}}
\newcommand{\imag}{\mathrm{Im}}
\newcommand{\CNOT}{\mathrm{CNOT}}

\newcommand{\Control}{\mathrm{Control}}

\newcommand{\Ord}[1]{\mathcal{O}\left( #1 \right)}
\newcommand{\tOrd}[1]{\widetilde{\mathcal{O}}\left( #1 \right)}
\newcommand{\Tht}[1]{\Theta \left( #1 \right)}

\newcommand{\ErrRel}{\varepsilon_{\rm rel} }

\def\be{\begin{eqnarray}}
\def\ee{\end{eqnarray}}

\usepackage{color}
\definecolor{Pr}{rgb}{0.4,0.3,0.9}

\definecolor{Lm}{rgb}{0.6,0.2,0.2}

\definecolor{cornsilk}{rgb}{1.0, 0.97, 0.86}
\definecolor{darkgray}{rgb}{0.66, 0.66, 0.66}
\definecolor{cordovan}{rgb}{0.54, 0.25, 0.27}

\begin{document}

\title{Quantum-classical algorithms for skewed linear systems with optimized Hadamard test}
\author{Bujiao Wu} 
\email{bujiaowu@gmail.com}
\affiliation{
Institute of Computing Technology, Chinese Academy of Sciences}
\affiliation{
University of Chinese Academy of Sciences}
\author{Maharshi Ray}
\affiliation{
Centre for Quantum Technologies, National University of Singapore}
\author{Liming Zhao}
\affiliation{
Centre for Quantum Technologies, National University of Singapore}
\author{Xiaoming Sun}
\affiliation{
Institute of Computing Technology, Chinese Academy of Sciences}
\affiliation{
University of Chinese Academy of Sciences}
\author{Patrick Rebentrost} 
\email{cqtfpr@nus.edu.sg}
\affiliation{
Centre for Quantum Technologies, National University of Singapore}

\date{\today}

\begin{abstract}
The solving of linear systems provides a rich area to investigate the use of nearer-term, noisy, intermediate-scale quantum computers.
In this work, we discuss hybrid quantum-classical algorithms for skewed linear systems for over-determined and under-determined cases. Our input model is such that the columns or rows of the matrix defining the linear system are given via quantum circuits of poly-logarithmic depth and the number of circuits is much smaller than their Hilbert space dimension.  
Our algorithms have poly-logarithmic dependence on the dimension and polynomial dependence in other natural quantities. 
In addition, we present an algorithm for the special case of a factorized linear system with run time poly-logarithmic in the respective dimensions. 
At the core of these algorithms is the Hadamard test and
in the second part of this paper we consider the optimization of the circuit depth of this test. 
Given an $n$-qubit and $d$-depth quantum circuit $\mathcal{C}$, we can approximate $\bra{0}\mathcal{C}\ket{0}$ using $(n + s)$ qubits and $O\left(\log s + d\log (n/s) + d\right)$-depth quantum circuits, where $s\leq n$.
In comparison, the standard implementation requires $n+1$ qubits and $O(dn)$ depth. 
Lattice geometries underlie recent quantum supremacy experiments with superconducting devices.
We also optimize the Hadamard test for an $(l_1\times l_2)$ lattice with $l_1 \times l_2 = n$, and can approximate $\bra{0}\mathcal{C}\ket{0}$ with $(n + 1)$ qubits and $O\left(d \left(l_1 + l_2\right)\right)$-depth circuits. In comparison, the standard depth is $O\left(d n^2\right)$ in this setting. Both of our optimization methods are asymptotically tight in the case of one-depth quantum circuits $\mathcal{C}$.
\end{abstract}

\maketitle

\section{Introduction\label{sec:intro}}

Linear systems appear in a variety of problems in engineering, the physical sciences, and machine learning. A commonly encountered situation is that, given a matrix $\Amat\in \Cbb^{N\times M}$, and a vector $\bm b\in \Cbb^{N}$, we would like to find an optimized solution $\bm x^*$ such that $\bm x^*=\argmin_{\bm x} \vabs{\Amat \bm x - \bm b}$, in $\ell_2$-norm.
Two important scenarios can be discussed. The simpler case is when the number of variables is much less than the number of equations.
An example is the common regression problem of fitting a low-degree polynomial to big data~\cite{draper1998applied,montgomery2012introduction}. 
The second case is when the number of equations  is much less than
the number of variables.
The linear system is under-determined and can have many solutions. Further constraints can help to find a unique solution. Popular scenarios are the LASSO estimation~\cite{tibshirani1996regression,santosa1986linear} and compressed sensing~\cite{candes2006robust,candes2008introduction} that arise in many applications in statistics and applied mathematics~\cite{efron1994introduction,weisberg2005applied}. 

Due to the wide range of applications, there are continued developments in algorithms, including classical algorithms~\cite{le2014powers,spielman2010algorithms,spielman2004nearly,cohen2014solving,andoni2018solving,van2020solving,strohmer2009randomized}, quantum algorithms~\cite{harrow2009quantum,ambainis2012variable,childs2017quantum,dervovic2018quantum,wossnig2018quantum,wang2017quantum,lin2019solving,shao2019randomized}, quantum-inspired classical algorithms~\cite{chia2018quantum,gilyen2018quantum}, and hybrid variational algorithms~\cite{huang2019near, mcardle2019variational,xu2019variational,bravo2019variational}.
For exact solving of general square matrices the latest classical algorithm is discussed by Gall ~\cite{le2014powers} (where $M = N$), which has running time $\Ord{N^{\omega}}$, where $\omega< 2.373$ is the matrix multiplication exponent. 
Cohen \emph{et al.}~\cite{cohen2014solving} proposed an approximate $\Ord{d\sqrt{\log N}\pbra{\log \log N}^{\Ord{1}} \log (1/\ErrRel) }$ time algorithm for $d$-sparse symmetric diagonally dominated matrices with relative error $\ErrRel $ in an optimized solution $\bm x^\ast$.
Subsequently, Andoni, Krauthgamer, and Pogrow~\cite{andoni2018solving} proposed an approximate sublinear time algorithm for symmetric diagonally dominated matrices. 
Recently, van den Brand \emph{et al.}~\cite{van2020solving} proposed an $\tOrd{NM + M^3}$ time randomized algorithm for $N\times M$ matrices with high probability. The notation $\tOrd{f(n)}$ means $\Ord{f(n)\polylog n}$.

Quantum computing promises speedups for several computationally hard problems ~\cite{aaronson2017complexity,shor1994algorithms, grover1996fast}. There are a number of quantum approaches to the linear systems problem. Harrow, Hassidim, and Lloyd~\cite{harrow2009quantum} (HHL) proposed the first quantum algorithm to solve $d$-sparse $N\times N$ linear systems with running time $\Ord{d^2\kappa^2 \log N /\varepsilon}$, where $\kappa$ is the condition number of matrix  $\Amat$, and $\varepsilon$ is the additive error. HHL algorithm does not output the optimal solution $\bm x^*$ but prepares a normalized quantum state $\ket{x}$ which is proportional to $\bm x^*$ approximately. The original HHL algorithm requires the use of phase estimation and sparse Hamiltonian simulation as subroutines, which is beyond the reach of present-day quantum computers for larger dimensions.
Several follow-up quantum algorithms~\cite{ambainis2012variable,childs2017quantum,dervovic2018quantum} improve the condition number, sparsity, and error dependence. 
Wossnig, Zhao, and Prakash~\cite{wossnig2018quantum} gave an $\Ord{\kappa^2\sqrt{N}\polylog N/\varepsilon}$ time quantum algorithm for dense linear systems with spectral norm bounded by a constant, and a column/row-based input model. 
Wang~\cite{wang2017quantum} proposed a $\poly\pbra{\log N, M, \kappa, 1/\varepsilon}$ time quantum algorithm to solve the linear regression problems for a full-rank matrix $\Amat$ ($\rank(\Amat) = M$).
Shao and Xiang~\cite{shao2019randomized} give a quantum version of the stochastic Kaczmarz algorithm~\cite{strohmer2009randomized} for linear systems in a column/row-based input model, with time polynomial in the logarithm of the number of columns/rows. 

Recent works also include other randomized classical algorithms for low-rank matrices, whose discussion was inspired by quantum algorithms especially in the context of machine learning. Chia, Lin and Wang~\cite{chia2018quantum} proposed an approximate algorithm with time complexity $\tOrd{\vabs{\Amat}^{6\omega + 4}\vabs{\Amat}_{F}^{2\omega}\kappa^{8\omega + 4}/\ErrRel^{2\omega}}$. 
The algorithm uses similar techniques as the dequantization algorithm for recommendation systems~\cite{kerenidis2017quantum} proposed by Tang~\cite{tang2019quantum}. 
The algorithms require that $\Amat$ and $\bm b$ are given by low-overhead data structures and entries of $\Amat$ and $\bm b$ can be sampled according to their size
and achieve a sublinear time complexity.
Gily{\'e}n, Lloyd, and Tang~\cite{gilyen2018quantum} have proposed a similar algorithm for low-rank systems with time complexity $\Ord{\vabs{\Amat}_{F}^{6}\kappa^{16}R^{6}\polylog  (N,M)/\ErrRel^{6}}$ and the same assumptions, where $R$ is the rank of matrix.
Although the time complexity is sublinear in the matrix dimension for the above two algorithms, they are only efficient for low-rank matrices and show substantial dependence on the relative error $\ErrRel$ and other parameters $\|\Amat\|_F, \kappa,\|\bm b\|$ and rank $R$. 
Note that recently Gilyen, Song, and Tang~\cite{Gilyen2020improved} proposed a dequantized algorithm with improved time complexity $\Ord{\vabs{\Amat}_F^6\vabs{\Amat}^2\vabs{\Amat^{-1}}^8/\ErrRel^4}$.

With the recent advent of small-size, noisy quantum computers~\cite{arute2019quantum,IBMQ,ye2019propagation}, Noisy Intermediate Scaled Quantum (NISQ) applications have received a significant amount of attention. A series of near-term variational quantum algorithms have been proposed~\cite{xu2019variational,bravo2019variational,mcardle2019variational,huang2019near} to solve certain linear system problems with square matrices. 
These algorithms use parameterized quantum circuits to construct Ansatz states and employ a quantum-classical loop to train the parameters via a suitable cost function.
The data input model of these works is one where the matrix is given as a linear combination of a relatively small number of unitaries.
In principle, these algorithms may allow to take advantage of NISQ devices and offer potential quantum advantages for certain linear systems. However, often they suffer from a barren plateau/local minima issue and do not theoretically guarantee the run time.
Huang, Bharti, and Rebentrost~\cite{huang2019near} discuss a linear combination of quantum states method to avoid the aforementioned barren plateau problem common in many NISQ algorithms.  In some cases,  many (proportional to the dimension) variational quantum states have to be combined to provably obtain the solution. 

In this work, we give algorithms for skewed linear systems for potential use on NISQ devices.
We consider a different input model compared to the linear combination of unitaries input of Refs.~\cite{xu2019variational,bravo2019variational,mcardle2019variational,huang2019near}. 
The columns of the matrix $\mathbf{A}\in \mathbb{C}^{N \times M}$ shall be given via quantum circuits of poly-logarithmic depth and with the assumption that $M \ll N$.  
The first algorithm is for solving the over-determined problem $\mathbf{A} \bm x = \bm b$, with vector $\bm b\in \mathbb{C}^{N}$ given by an efficient quantum circuit. 
The second algorithm is for solving the under-determined problem
$\mathbf{A}^\dagger \bm y = \bm c$, with vector $\bm c\in \mathbb{C}^{M}$ given classically.
The input model can enhance the types of linear system problems that could be solved on a near-term quantum computer. Our model can be in principle converted to the linear combination of unitaries input model, with an additional cost which may go beyond the NISQ setting, see
Appendix \ref{app:CompareInput} for the details.
Under the given assumptions, our algorithms approximately finds the true solution with high probability. 
In both cases, the run time is polynomial in the dimension $M$ and the circuit depth of the unitaries defining the input model, \emph{i.e.}, ${\rm poly} \log N$.
We also consider factorized linear systems. Instead of accessing the matrix $\Amat$, we can only access the rank $R$ matrices $\Amat_1\in \Cbb^{N\times R}, \Amat_2\in \Cbb^{R\times M}$ such that $\Amat = \Amat_1 \Amat_2$. This setting  has several applications in machine learning~\cite{kim2008sparse,zou2006sparse}.
Ma, Needell, and Ramadas~\cite{ma2018iterative} proposed a classical stochastic iterative algorithms for factorized linear systems.
Their algorithms have an exponentially fast convergence for some types of linear systems, while the run time per step depends on whether $\Amat$ is consistent or not.
Our algorithm solves such systems with time polynomial in the rank $R$ and logarithmic in the dimensions $N,M$.

We list the comparison of our algorithms, the well known stochastic gradient decent algorithm~\cite{strohmer2009randomized} and the recent dequantized algorithms~\cite{chia2018quantum, gilyen2018quantum} in Table \ref{tab:ComparisonAlg}. 
As discussed, these algorithms have different input models.

The Hadamard test is the central subroutine in algorithms for near-term linear systems \cite{huang2019near,bravo2019variational,xu2019variational} and also the present work.
In short, given a unitary $\Umat$, the Hadamard test of $\{\ket{0},\Umat\}$\footnote{We use $\ket{0}$ to represent $\ket{0^t}$ where $\Umat\in \Cbb^{2^t}$.} is to calculate value the $\bra{0}\Umat\ket{0}$ by independently and repeatedly performing a quantum circuit related to $\Umat$.
In the second part of the work, we discuss circuit optimization of the Hadamard test. We specifically consider the optimization of the Hadamard test when the qubits have (i) all-to-all connectivity and (ii) 2D nearest-neighbor connectivity as in recent 2D quantum devices \cite{arute2019quantum,IBMQ}.

Our results for linear systems can be informally described with the following theorems.
\begin{theorem}[\emph{Informal}]
For an $N\times M$ ($M\times N$) dimensional matrix for which the columns (rows) can be accessed in quantum form with $\Ord{\log N}$ qubits and corresponding right-hand side vector, precisely defined in Assumptions \ref{assume:quantummatrix}, \ref{assume:vecb} and Problems \ref{pro:overdetermined}, \ref{pro:underdetermined} below,
there exists a hybrid algorithm with run time
polynomial in $M\log N,$ and dependent on $1/\varepsilon^4$, such that the algorithm outputs classical data to construct the solution to accuracy $\varepsilon$ with high success probability.
\label{thm:HybridAlgLogMIntro}
\end{theorem}
We postpone the formal statement of this theorem to Theorems \ref{thm:SecOverdetermined} and \ref{thm:SecUnderdetermined} in Sec. \ref{sec:Hybrid Alg}. In addition, we obtain the following informal theorem on factorized linear systems.

\begin{theorem}[\emph{Informal}]
Suppose we are given a rank $R$ matrix of dimension $N \times M$ in terms of the product of two matrices of dimension $N \times R$ and  $R \times M$. The first matrix can be accessed in column form via quantum circuits over $\Ord{\log N}$ qubits and the second matrix can be accessed in row form via quantum circuits over $\Ord{\log M}$ qubits. See Problem \ref{pro:lowrank} and Assumptions \ref{assume:quantummatrix}, \ref{assume:vecb} for the precise definitions.
There exists a hybrid algorithm with run time polynomial in the rank $R$, polynomial in ${\log  N}$ and ${\log M}$, and dependent on $1/\varepsilon^2$, which outputs classical data such that an  $\varepsilon$-approximate solution can be constructed with high success probability.
\label{thm:LowRank}
 \end{theorem}
The algorithm of this problem, the formal statement and the proof are postponed into Appendix \ref{app:FactorizedLS}.
We also give a similar result for the following ``rank-relaxed" factorized linear system. We relax the rank of $\Amat$ and $\Amat_1$ to be less than $R$, but keep the rank of $\Amat_2$ at $R$, and access $\Amat_1, \Amat_2$ and $\bm b$ with the same assumptions as in Theorem \ref{thm:LowRank}, and obtain similar run time and correctness guarantees. 
The formal statement and proof are in Appendix \ref{app:GeneralFactorizedLS}.

Since the Hadamard test is the main subroutine in the above algorithms, we consider the problem of optimizing the depth of the quantum circuit implementing the Hadamard test.
In the second part of this paper, we first optimize the circuit of the Hadamard test for qubits having an all-to-all connectivity (two-qubit gates can be applied on any two qubits) with limited ancillas, as stated in the following theorem.

\begin{theorem}
Suppose we are given an $n$-qubit unitary $\Umat$ which has a $d$-depth circuit $\Ccal$. There exists an $n+s$-qubit and $\Ord{\log s + d\log (n/s) + d}$-depth circuit for Hadamard test of $\cbra{\ket{0^n}, \Umat}$, where $s\in[n]$.
\label{thm:sancillasIntro}
\end{theorem}

NISQ devices only have restricted  connectivity of  qubits. The connections of the qubits can be represented as a graph, and two-qubit gates can only operate on two qubits which are connected in the graph.
There is a series of recent works mapping unitary circuits without constraints to quantum devices with some connectivity graphs~\cite{kissinger2020cnot,nash2020quantum,wu2019optimization,itoko2020optimization}. 
The graph for the existing quantum devices~\cite{arute2019quantum,IBMQ,ye2019propagation} is in all cases the planar graph. Hence, we consider the optimization of Hadamard test on such a graph, as stated in the following theorem.

\begin{theorem}
Suppose we are given an $n$-qubit unitary $\Umat$ which has a $d$-depth circuit $\Ccal$ and we have to re-design it such that the qubits are constrained according to an $l_1 \times l_2$ lattice, where $l_1l_2 = n$. There exists 
an $\Ord{d\pbra{l_1 + l_2}}$-depth circuit for Hadamard test of $\cbra{\ket{0^n}, \Umat}$ under the above lattice.
\label{thm:2DControlUIntro}
\end{theorem}

Theorem \ref{thm:2DControlUIntro} is obtained by mapping the circuit of Theorem \ref{thm:sancillasIntro} to the lattice. We also generalize it to any connected graph with a Hamiltonian path, as stated in the following corollary.

\begin{corollary}
Suppose we are given an $n$-qubit unitary $\Umat$ which has a $d$-depth circuit $\Ccal$ and we have to re-design it such that the qubits are constrained according to a graph $G$ which has a Hamiltonian path. There exists an $\Ord{dn}$-depth circuit for Hadamard test of $\cbra{\ket{0^n}, \Umat}$ under this graph.
\label{coro:TopoAnyGCtrlUIntro}
\end{corollary}

We also give a lower bound for the quantum depth of Hadamard test in Theorem \ref{thm:LowerboundHadamardIntro}.
 By Theorem \ref{thm:LowerboundHadamardIntro} when $d = \Ord{1}$, the optimized depth of Hadamard test in Theorems \ref{thm:sancillasIntro}, \ref{thm:2DControlUIntro}, and Corollary \ref{coro:TopoAnyGCtrlUIntro} are all optimal.
\begin{theorem}
There exists an $n$-qubit unitary $\Umat\in \Cbb^{2^n\times 2^n}$ such that for any quantum device under the graph $G$ which has diameter\footnote{Diameter of the graph $G(V,E)$: the maximum distance of any two vertices in $G$, where the distance of two vertices is the minimum path connected these two vertices.} $D$, there needs at least $\Omega(\max\cbra{\log n, D})$-depth quantum circuit to generate $\bra{0^{n}}\Umat\ket{0^{n}}$ under graph $G$.
\label{thm:LowerboundHadamardIntro}
\end{theorem}

The rest of the paper is organized as follows. Sec.~\ref{sec:Preliminaries} introduces notations, assumptions, and definitions.
Sec.~\ref{sec:Hybrid Alg} introduces the hybrid algorithms for linear systems with skewed dimensions.
Sec.~\ref{sec:CirHadamardTest} introduces the circuit of the Hadamard test and give two depth-optimal algorithms for the Hadamard test, for the cases of all-to-all connectivity and 2D nearest-neighbor connectivity, respectively.
Sec.~\ref{sec:discussion} concludes with a discussion of open problems and directions for future work.

\section{Preliminaries\label{sec:Preliminaries}}
We define notation and linear algebra basics.
The quantum state $\ket{a}:=\frac{\bm a}{\vabs{\bm a}}$ represents the  vector $\bm a$ with normalization $\vabs{\bm a}$ in $\ell_2$-norm. 
The set of integers in a range is denoted as $[n]=\{1,2,\cdots, n\}$. 
Let $\Amat^{\dagger}$ denote the conjugate transpose of matrix $\Amat$ and $\kappa(\Amat) =: \kappa$ denote the condition number of $\Amat$ which is the ratio of the largest to smallest singular value of $\Amat$.
$N(\Amat)$ is the null space of $\Amat$, and $P_{L}$ is the orthogonal projector on space $L$.
The matrix $\Amat\in \Cbb^{N\times M}$, can be written as $\Amat := \Umat \bm \Sigma \Vmat^{\dagger}$, where $\Umat\in\Cbb^{N\times N}, \Vmat\in \Cbb^{M\times M}$ are two unitary matrices, and $\bm \Sigma \in \Cbb^{N\times M}$ is a rectangular diagonal matrix with entries $\geq 0$. Define $\Amat^{-1} := \Vmat \bm\Sigma^{-1} \Umat^{\dagger}$ be the pseudo-inverse of matrix $\Amat$, where $\bm \Sigma^{-1}$ is formed by replacing every  non-zero diagonal element by its reciprocal. 
Here, $\vabs{\Amat}$ denotes the $\ell_2$-norm of $\Amat$, and it equals the maximum singular value of $\Amat$. Hence $\vabs{\Amat}=\vabs{\Amat^\dagger}$, and $\vabs{\Amat}\vabs{\Amat^{-1}}=\kappa(\Amat)$.
For any matrix $\Amat\in \Cbb^{N\times M}$ and vector $\bm b\in \Cbb^{N}$, $\vabs{\Amat \bm b}\leq \vabs{\Amat}\vabs{\bm b}$.
We call a linear system $\Amat \bm x = \bm b$ a \emph{consistent linear system} if $\Amat\Amat^{-1}\bm b = \bm b$, otherwise we call it an \emph{inconsistent linear system}.
When we use $\log N$ to represent the number of qubits, it represents the smallest integer which is greater than $\log N$ with a little abuse of symbols.

In the quantum setting for solving linear systems,
we make several assumptions for the quantum access the matrix $\Amat$. Each normalized column of the  matrix $\Amat\in \Cbb^{N\times M}$ shall be given by a $\poly (\log N)$-depth quantum circuit, and the $\ell_2$-norms of each column shall be given classically.
\begin{assumption}
For $j\in [M]$, assume knowledge of norms
$\vabs{\bm a_{j}}\geq 0$ and 
 quantum circuits $\Umat_{j} \in \Cbb^{N \times N}$ such that $\Umat_{j} \ket{ 0^{\log N}} =: \ket{ a_{j}}\in\Cbb^{N}$.
 The unitaries $\Umat_{j} $ shall have a circuit depth of at most $\poly(\log N)$.
The vectors $\ket{ a_{j}}$ and norms define the columns of a matrix 
$\Amat := \sum_{1\leq j \leq M} \vabs{\bm a_{j}} \ket{ a_j}\bra{j}\in \Cbb^{N\times M}$. Similarly, the vectors $\ket{ a_{j}}$ and norms define the rows of the matrix $\Amat^\dagger \in \Cbb^{M\times N}$.
\label{assume:quantummatrix}
\end{assumption}

Next, for the case of a linear system $\Amat \bm x = \bm b$, we assume quantum access a vector $\bm b\in \Cbb^N$ with a $\polylog N$  depth quantum circuit. For the linear system 
$\Amat^\dagger \bm y = \bm c$, we only require classical access the right-hand side. 
\begin{assumption}
We are given a unitary $\Umat_{\bm b}$ such that $\Umat_{\bm b} \ket{ 0^{\log N}} = \ket{ b}=\frac{\bm b}{\vabs{\bm b}}$, and the knowledge of the norm of $\vabs{\bm b} > 0$. 
The unitary shall have a circuit depth of at most $\polylog N$.
\label{assume:vecb}
\end{assumption}

We now define the problems investigated in this work. 
We focus on solving linear systems with the above assumptions. The first problem is associated with the over-determined case when there are less variables than constraints.

\begin{problem}
Given matrix $\Amat \in \mathbb{C}^{N\times M}$ according to Assumption \ref{assume:quantummatrix}, vector $\bm b\in \mathbb{C}^{N}$ according to Assumption \ref{assume:vecb}, and  $\varepsilon>0$, find an approximation $\hat{\bm x}\in \mathbb{C}^{M}$ of the optimal solution such that $\vabs{\Amat \hat{\bm x} - \bm{b}} -\min_{\bm x}\vabs{\Amat \bm x - \bm{b}}\leq \varepsilon$.
\label{pro:overdetermined}
\end{problem}
The second problem is associated with the under-determined case, when there are more variables than constraints. 
\begin{problem}
Given matrix $\Amat \in \mathbb{C}^{N\times M}$ according to Assumption \ref{assume:quantummatrix}, vector $\bm c\in \mathbb{C}^{M}$ given classically, and  $\varepsilon>0$, find an approximation $\hat{\bm y}\in \mathbb{C}^{N}$ of the optimal solution such that $\vabs{\Amat^\dagger  \hat{\bm y} - \bm{c}} -\min_{\bm y}\vabs{\Amat^\dagger \bm y - \bm{c}}\leq \varepsilon$. 
\label{pro:underdetermined}
\end{problem}

Both Problems \ref{pro:overdetermined} and \ref{pro:underdetermined}  can be solved in polylog time in $N$ under some reasonable assumptions. 
Next, we introduce a factorized linear system in which the matrix is given as a product of two matrices. The problem is defined as follows.
\begin{problem}[Factorized linear system]
Let $N,M \geq R$. Let matrix $\Amat \in \mathbb{C}^{N\times M}$ be a rank $R$ matrix such that 
$\Amat = \Amat_{1} \Amat_{2}$. The rank $R$ matrices
$\Amat_{1}^{\dagger} \in \Cbb^{R\times M}$ and $\Amat_{2} \in \Cbb^{R\times N}$ are given according to Assumption \ref{assume:quantummatrix}, vector $\bm{b}\in \mathbb{C}^{N}$ according to Assumption \ref{assume:vecb}. Find an approximation $\hat{\bm x}\in \mathbb{C}^{N}$ such that $\vabs{\Amat\hat{\bm x} - \bm b} - \min_{\bm x}\vabs{\Amat \bm x - \bm{b}} \leq \varepsilon$.
\label{pro:lowrank}
\end{problem}

The following consistent linear system serves for the intermediate process of solving Problems \ref{pro:overdetermined} and \ref{pro:underdetermined}. Specifically, each element of the square matrix $\Vmat$ in the following linear system is related to the expectation of a sum of Bernoulli random variables.

\begin{problem}
\label{pro:QuadProApp}
Let $\Delta >0$. Let $\Vmat\in\Cbb^{K\times K}$ be a positive semi-definite matrix and $\bm q\in \Cbb^{K}$ be a vector, where the entries are a product of a bounded complex number and a bounded real number. Formally, the matrix is defined by $V_{ij} = \xi_{ij}v_{ij}$, where $v_{ij} \in \Cbb$ with $\abs{v_{ij}}\leq 1$ and $\xi_{ij} \in \Rbb$ with $| \xi_{ij} |\leq \Delta$.
The vector is defined by
$q_{i} = \nu_{i}u_{i}$, where 
$u_{i} \in \Cbb$ with
$\abs{u_{i}}\leq 1$ and 
$\nu_{i} \in \Rbb$ with $| \nu_{i} |\leq \Delta$.
Let $\hat{v}_{ij},\hat{u}_i$ be random variables with expectation values $v_{ij} =\mathbb E[\hat v_{ij}]$ and $u_i=\mathbb E[\hat u_i]$ respectively.
The real and imaginary parts of $\hat{v}_{ij}$ and $\hat{u}_i$ have the form $2(B_1+\cdots + B_S)/S - 1$ for some particular $S$ independent and identical Bernoulli trials $B_1,\dots, B_S$.
These random variables define the random matrix $\hat{\Vmat} \in \Cbb^{K\times K}$
with elements $\hat{V}_{ij} = \xi_{ij}\hat{v}_{ij}$ and the random vector $\hat{\bm q} \in \Cbb^{K}$ with elements 
$\hat{q}_i = \nu_i \hat{u}_i$, for which $\Vmat = \mathbb E[\hat \Vmat]$ and ${\bm q} = \mathbb E[\hat{\bm q}]$, respectively. Let $\bm \alpha^\ast := \argmin_{\bm \alpha}\Vert \Vmat \bm \alpha - \bm q \Vert$.
With these definitions, the problem statement is as follows.
Given all $\xi_{ij}$  and $\nu_i$,  sampling access to  $\hat v_{ij}$ and $\hat u_i$ via the Bernoulli trials, $\eta>0$, and $\Delta\geq 0$, produce a vector $ \hat{\bm \alpha}$ such that $\vabs{\hat{\bm \alpha} - \bm \alpha^\ast}\leq \eta$ with high probability. 
\end{problem}

\section{Hybrid algorithms for skewed linear systems\label{sec:Hybrid Alg}}
In this section, we give two hybrid algorithms for Problems \ref{pro:overdetermined} and \ref{pro:underdetermined} with Assumptions \ref{assume:quantummatrix} and \ref{assume:vecb}.
We first give technical lemmas in the following subsection.

\subsection{Technical lemmas for over- and under-determined linear systems\label{subsec:Technical}}

The following lemmas show that when we slightly perturb the diagonal elements of the matrix of a consistent linear system, the difference of the solutions of these two systems is bounded. This fact is useful when proving the error bounds for the algorithms in this work.

\begin{lemma}
Let a positive semi-definite matrix $\Vmat\in \Cbb^{K\times K}$ and vector $\bm q\in \Cbb^{K}$ be such that $\Vmat^{-1}\Vmat \bm q = \bm q$. 
Also let the shifted matrix be $\Wmat := \Vmat + \lambda \Imat$, where $\lambda > 0$. 
The Euclidean distance between the pseudo-inverse solutions of these two linear systems, that is 
$\vabs{\Vmat^{-1}\bm q - \Wmat^{-1}\bm q}$, is upper-bounded by 
$\lambda \vabs{\Vmat^{-1}}^2  \vabs{\bm q}$.
\label{lem:PerturbBounded}
\end{lemma}
\begin{proof}
Let the eigen-decomposition of $\Vmat$ be $\Vmat := \Umat \bm\Sigma \Umat^{\dagger}$, where $\Umat$ is a unitary and
$\bm\Sigma :=\diag\pbra{\sigma_1,\cdots, \sigma_P, 0,\cdots, 0}$ is a diagonal matrix, where $\sigma_1\geq \sigma_2\geq \cdots \geq \sigma_P>0$, and $P\leq K$, then we have $\Wmat = \Umat \pbra{\bm\Sigma + \lambda \Imat} \Umat^{\dagger}$. 
In the matrix $\bm\Sigma + \lambda \Imat$ all diagonal elements are shifted. In contrast, let $\bm \Sigma'= \diag\pbra{\sigma_1 + \lambda, \cdots, \sigma_P+\lambda, 0,\cdots, 0} \in \Cbb^{K\times K}$ 
be a diagonal matrix such that
only the non-zero elements are shifted.
Since $\Vmat^{-1}\Vmat \bm q = \bm q$, $\bm q$ is in the space spanned by eigenvectors of $\Vmat$ with non-zero eigenvalues. Hence $\Wmat^{-1}\bm q = \Umat \bm\Sigma'^{-1}\Umat^{\dagger}\bm q$, where $\bm\Sigma'^{-1}$ is the pseudo-inverse of $\bm\Sigma'$. For the difference in solutions we obtain 
\be
  \vabs{\Vmat^{-1}\bm q - \Wmat^{-1}\bm q} &=& \vabs{\Umat\bm\Sigma^{-1}\Umat^{\dagger}\bm q - \Umat\bm\Sigma'^{-1}\Umat^{\dagger}\bm q}\\
&\leq& \vabs{\bm\Sigma^{-1} - \bm\Sigma'^{-1}}\vabs{\bm q}\\
&=& \pbra{\frac{1}{\sigma_P} - \frac{1}{\sigma_P + \lambda}}\vabs{\bm q}\\
&\leq& \lambda\vabs{\Vmat^{-1}}^2\vabs{\bm q}.
\ee
\end{proof}

Consider the solving of approximate linear systems. 
For a linear system with matrix $\Vmat\in \Cbb^{K\times K}$ and vector $\bm q \in\Cbb^K$, let
the approximate matrix be $\hat{\Vmat}$ and the approximate vector be $\hat{\bm q}$.
These approximate quantities could be obtained from the measurements outputs of a suitable quantum computation (here the Hadamard test).
The next lemma shows a concentration bound for such random matrices.

\begin{lemma} \label{lem:randommatrix}
Let $\Vmat\in\Rbb^{K\times K}$ be a real symmetric matrix. 
As in Problem \ref{pro:QuadProApp} restricted to real numbers, let the absolute value of all elements of $\Vmat$ be less than $\Delta >0$ and $\hat{\Vmat}$ be as defined. With $\delta >0$, let $S \log 2/\delta$ be the number of samples for each element of $\hat \Vmat$. Then,
$\vabs{\hat{\Vmat}-\Vmat}= \Ord{\Delta\sqrt{K/S}}$ with high probability.
\end{lemma}
\begin{proof}
Recall the definition of $\Vmat$ and $\hat \Vmat$ in Problem \ref{pro:QuadProApp}.
Let $\Xmat :=\pbra{ \hat{\Vmat} - \Vmat}\in \Rbb^{K\times K}$. 
Let $\hat{v}_{ij}=2\pbra{B_1+\cdots + B_S}/S-1$, where $B_1, \dots, B_S$ are $S$ independent and identical samplings of Bernoulli trials.
Then $\hat v_{ij}$ multiplied by the scalar $\xi_{ij}$ gives a random variable $\hat{\Vmat }_{ij}$ with variance $\Ord{\Delta^2/S}$.
Hence, 
\be\Ebb[X_{ij}^{2}] = \Ebb\left[\pbra{\hat{\Vmat}_{ij} - \Vmat_{ij}}^{2}\right] = {\rm Var}(\hat{\Vmat}_{ij}) = \Ord{\Delta^2/S}.
\ee 
Using $S\log 2/\delta$ samples per entry, we also have Hoeffding's inequality
for all $i,j$ for the failure probability 
\be
\Pmat[\abs{\hat{\Vmat}_{ij} - \Vmat_{ij}} \geq t ] \leq 2 e^{- 2S t^2 \log(2/\delta) /\Delta^2}.
\ee
Set $t = \Delta/  \sqrt{2 S}$. Then
\be
\Pmat[\abs{\hat{\Vmat}_{ij} - \Vmat_{ij}} \geq \Delta / \sqrt{2S} ] &\leq& \delta.
\ee
Hence the corresponding success probability is at least $1-\delta$.

To align with the random matrix results in Ref.~\cite{bandeira2016sharp}, define 
$\tilde{\sigma}:= \max_{i} \sqrt{\sum_{j}\Ebb[X_{ij}^{2}]}$ which is bounded as
$\tilde{\sigma}=\Ord{\Delta\sqrt{K/S}}$.
Also consider
the infinity norm 
$\vabs{X_{ij}}_{\infty}$. Using the Hoeffding bound above, we can evaluate
\be
\vabs{X_{ij}}_{\infty} &=& \lim_{p\to \infty} 
\Ebb\left[\abs{X_{ij}}^{p}\right]^{1/p}\\
&=&\lim_{p\to \infty} 
\Ebb\left[\abs{\hat{\Vmat}_{ij} - \Vmat_{ij}}^{p}\right]^{1/p}\\
&\leq& 
\lim_{p\to \infty} 
\Ebb\left[\left (\frac{\Delta}{\sqrt{2S}}\right )^{p}\right]^{1/p} = \frac{\Delta}{\sqrt{2S}},
\ee 
which holds with high probability. 
The random matrix results use  the variable $\tilde\sigma_{*} := \max_{ij}\vabs{X_{ij}}_{\infty}
$, which can be bounded as  $\tilde\sigma_{*}\leq \Delta/\sqrt {2S} $ with high probability.
By Corollary 3.12 in Ref.~\cite{bandeira2016sharp}, there exists for any $0< \varepsilon \leq 1/2$ a universal constant $c_{\varepsilon}$ such that for every $t\geq 0$,
\be
\Pmat[\vabs{\Xmat}\geq (1 + \varepsilon)2 \tilde{\sigma} + t]\leq K e^{-t^{2}/(c_{\varepsilon}\tilde{\sigma}_{*}^{2})}.
\ee
Set $\varepsilon = 1/2$ and 
$t = \Delta\sqrt{{c_{\varepsilon}K/S}}$. 
Then we have
\be
\Pmat\left[\vabs{\hat{\Vmat} - \Vmat}\geq c_{1}\Delta\sqrt{K/S} \right]\leq Ke^{-2K},
\ee
for a suitable constant $c_1$. 
The total success probability 
is hence at least
$(1-\delta)(1-e^{-2 K + \log K})$.
\end{proof}

We omit the $\log 1/\delta$ dependency in the remainder of this work. 
The following lemma shows how to bound the error of the approximate solution when we only have the approximate system. This lemma is a generalization of Proposition 10 in Ref.~\cite{huang2019near}, which requires the matrix $\Vmat$ to be invertible and cannot be used directly for low-rank positive \emph{semi}-definite matrices. 

Let $\bm \alpha^{\ast}$ be the optimal solution of the original system and $\hat{\bm \alpha}$ be the optimal solution of the approximate system.
We generalize the proof of  Ref.~\cite{huang2019near} by introducing the perturbed system $\Wmat = \Vmat + \lambda\Imat$ for a small perturbation $\lambda$, and prove that the error between the approximate system after perturbation $\hat{\Wmat} = \hat{\Vmat} + \lambda \Imat$ and the perturbed system $\Wmat$ is bounded. By Lemma \ref{lem:PerturbBounded}, the error between perturbed system $\Wmat$ and original system $\Vmat$ is also bounded. Hence the total error is bounded (error between systems $\hat{\Wmat}$ and $\Vmat$).
Note that in the following result we can use $\vabs{\bm \alpha^*}\leq \vabs{\Vmat^{-1}}\vabs{\bm q}$ as an upper bound in the absence of further knowledge about $\vabs{\bm \alpha^*}$.

\begin{lemma}
Let $\Vmat\in\Cbb^{K\times K}$ be Hermitian positive semi-definite and $\bm q\in \Cbb^{K}$ such that $\Vmat^{-1}\Vmat \bm q = \bm q$.
As in Problem \ref{pro:QuadProApp}, let the absolute value of all elements of $\Vmat$ and $\bm q$ be less than $\Delta >0$, and  $\hat{\Vmat}$, $\hat{\bm q}$, and $\eta$ be as defined.
Then we can solve Problem \ref{pro:QuadProApp} with the optimal solution $ \hat{ \bm \beta} := \pbra{\hat{\Vmat} + \lambda \Imat}^{-1}\hat {\bm q}$ with high probability, if the number of samples used for obtaining each entry of $\hat{\Vmat}$ and $\hat{\bm q}$ is $\Delta^2 T$, where
$\lambda \leq \eta\big/\pbra{2\vabs{\Vmat^{-1}}^2\vabs{\bm q}}$ and
$T = \Ord{\frac{K \pbra{\eta+ \vabs{\bm \alpha^\ast} + 1 }^2}{\lambda^2\eta^2}}$.
\label{lem:approxCorrect}
\end{lemma}
\begin{proof}
The complex numbers in $\hat {\Vmat}$ and $\hat{\bm q}$ can be  handled correctly by treating real and imaginary parts independently.
By Lemma \ref{lem:randommatrix} with $S = \lceil \Delta^2 T \rceil$, we have $\vabs{\hat{\Vmat}-\Vmat}= \Ord{\sqrt{K/T}}$ in $\ell_2$ norm with high probability.  
Similarly, with the same number of samples for each entry, we can also obtain an estimate $\hat{\bm q}$ for $\bm q$ such that $\vabs{\hat{\bm q}-\bm q}= \Ord{\sqrt{K/T}}$ in $\ell_2$ norm with high probability.

Let $\lambda >0$ and the shifted matrix be $\Wmat := \Vmat + \lambda \Imat$ with the solution
$\bm\beta := \Wmat^{-1} \bm q$. 
Also the shifted estimated matrix is $\hat{\bm \Wmat} := {\hat{\Vmat}} + \lambda \Imat$ with solution  $\hat{\bm \beta} := \hat{\Wmat}^{-1}\hat{\bm q}$, using the estimated vector $\hat {\bm q}$. Then we have $\vabs{\hat{\Wmat} - \Wmat}=\vabs{\hat{\Vmat} - \Vmat }= \Ord{\sqrt{K/T}}$. We also obtain $\vabs{\Wmat^{-1}}\leq 1/\lambda$ 
by the definition of $\Wmat$.
By Lemma \ref{lem:PerturbBounded}, the difference of exact solution and shifted exact solution is $\vabs{\bm \beta - \bm \alpha^*} \leq \lambda\vabs{\Vmat^{-1}}^2 \vabs{\bm q}$.
Using $\lambda \leq \eta /\pbra{2\vabs{\Vmat^{-1}}^2 \vabs{\bm q}}$ we obtain
$\vabs{\bm \beta - \bm \alpha^*} \leq \eta/2$.

In the next step, we prove that the distance between shifted exact solution and shifted approximate solution is bounded as $\vabs{{\bm \beta}- \hat{\bm \beta}}\leq\eta/2$.
Combining these bounds will  thus bound the difference between exact solution and shifted approximate solution as $\vabs{\bm{\alpha}^{*} - \hat{\bm \beta} }\leq \eta$.

Recall that 
$\bm \beta = \Wmat^{-1} \bm q,\hat{\bm \beta} = \hat{\Wmat}^{-1}\hat{\bm q}$ and the simple fact that $\Wmat^{-1} \Wmat = \Imat$. Hence, 
\be
 \hat{\bm \beta}- \bm\beta
&=& \Wmat^{-1}
\pbra{ \Wmat\hat{\bm \beta} - \Wmat \bm\beta} \\
&=& \Wmat^{-1} \pbra{\pbra{\Wmat - \hat{\Wmat}}\hat{\bm \beta} + 
\hat{\Wmat}\hat{\bm \beta}- \Wmat\bm\beta}\\
&=& \Wmat^{-1} \pbra{\Wmat - \hat{\Wmat}}\pbra{\hat{\bm \beta}- \bm\beta}  \\
&+& \Wmat^{-1}\pbra{\pbra{\Wmat - \hat{\Wmat}}\bm\beta + 
\hat{\bm q} - \bm q}.
\ee
We hence obtain for the distance that
\begin{align}
    \vabs{\hat{\bm \beta}- \bm\beta}&\leq
\vabs{\Wmat^{-1}}\vabs{\hat{\Wmat} - \Wmat}\vabs{\hat{\bm \beta}- \bm\beta} +\\
& \vabs{\Wmat^{-1}}\pbra{\vabs{\hat{\Wmat} - \Wmat}\vabs{\bm\beta} +\vabs{\hat{\bm q} - \bm q}}.
\end{align}
Hence,
\begin{align}
\begin{aligned}
       \vabs{\hat{\bm \beta}- \bm\beta}&\leq 
\frac{\vabs{\Wmat^{-1}}\pbra{\vabs{\hat{\Wmat} - \Wmat}\vabs{\bm\beta}+\vabs{\hat{\bm q} - \bm q}}}{{1- \vabs{\Wmat^{-1}}\vabs{\hat{\Wmat} - \Wmat}}}\\
&\leq \eta/2.
\end{aligned}
\label{eq:PartsolBound}
\end{align}
The last inequality holds when 
\be
\vabs{\Wmat^{-1}}\pbra{\vabs{\hat{\Wmat} - \Wmat}\vabs{\bm\beta}+\vabs{\hat{\bm q} - \bm q}}&\leq& \eta/4, \\ 
1- \vabs{\Wmat^{-1}}\vabs{\hat{\Wmat} - \Wmat}&\geq& 1/2.
\ee
These inequalities can be achieved as follows. 
We have
\be
  \vabs{\bm \beta} &=& \vabs{\bm \beta - \bm \alpha^{*} + \bm \alpha^{*}}
  \leq \vabs{\bm \beta - \bm \alpha^{*}} + \vabs{\bm \alpha^{*}}\\
  &\leq& \eta/2 + \vabs{\bm \alpha^{*}},
\ee
where the last inequality holds because of Lemma \ref{lem:PerturbBounded}.
In addition, take the number of samples controlling $\vabs{\hat{\Wmat} - \Wmat}$ and $\vabs{\hat{\bm q}-\bm q}$ to be 
\be
T=\frac{c K \vabs{\Wmat^{-1}}^2\pbra{\eta/2 + \vabs{\bm \alpha^\ast} + 1}^2}{\eta^2},
\ee
for a suitable constant $c$.
We can rephrase this expression as 
\be
T = \frac{c'K \pbra{\eta+ \vabs{\bm \alpha^\ast} + 1 }^2}{\lambda^2\eta^2},
\ee
with a suitable constant $c'$, using 
$\vabs{\Wmat^{-1}}\leq 1/\lambda$, as before. 
\end{proof}

\subsection{Hybrid algorithm for over-determined linear systems\label{subsec:overdetermined}}

Using the preceding lemmas, we describe our Algorithm \ref{Alg:overdetermined}, which gives a solution of Problem \ref{pro:overdetermined}. 
The time complexity is
$\Ord{ M^3\poly(\Gamma_1,\kappa,\vabs{\bm b},\log N)/\varepsilon^4}$, where $\Gamma_1 = \max_{j} \cbra{\vabs{\bm a_j}\vabs{\bm b}, \vabs{\bm a_j}^2}$, and $\bm a_j$ is the $j$-th column of matrix $\Amat$.

\SetNlSty{textbf}{}{ }
\begin{figure*}
\begin{center}
\begin{minipage}{1\linewidth}
\begin{algorithm}[H]
\label{Alg:overdetermined}
\SetKwInOut{Input}{Input}
\SetKwInOut{Output}{Output}
\Input{$\Amat\in \Cbb^{N\times M}$ that satisfies Assumption \ref{assume:quantummatrix}, and $\bm b\in \Cbb^{N}$ that satisfies Assumption \ref{assume:vecb}, error $\varepsilon > 0$.}
\Output{$\hat{\bm x}$ which is a solution of Problem \ref{pro:overdetermined}.}
Define matrix $\hat{\Wmat}\in \Cbb^{M\times M}$, and vector $\hat{\bm q}\in\Cbb^{M}$\;
Let $\Gamma_1 := \max_j\{ \vabs{\bm a_j}\vabs{\bm b}, \vabs{\bm a_j}^2\}$, 
$\lambda :=\varepsilon/\pbra{2\vabs{\Amat}^2\vabs{\Amat^{-1}}^4\vabs{\bm b}}$ and 
$T := \Ord{M \vabs{\Amat^{-1}}^4\kappa(\Amat)^4\vabs{\bm b}^2\pbra{\vabs{\Amat}(\vabs{\bm x^*} + 1)+\varepsilon}^2 \big/\varepsilon^4}$\;
Let $v_{jk} $ be the approximate value of $\braket{a_j|a_k}$ by repeatedly and independently measuring quantum circuit of Hadamard test $\cbra{ \ket{0^{\log N}}, \Umat_j^\dagger\Umat_k}$ $\Gamma_1^2 T$ times for all $j,k\in[M]$\;
Let the $(j,k)$-th element of $\hat{\Wmat}$  be $\hat{W}_{jk} := \vabs{\bm a_j} \vabs{\bm a_k} v_{jk}$ for any $j,k\in[M]$ and $j\ne k$ and  $\hat{W}_{jj} := \vabs{\bm a_j}^2 v_{jj} + \lambda$ for any $j\in[M]$\; 
Let $\hat{u}_j $ be the approximate value of $\braket{a_j|b}$ by repeatedly and independently measuring quantum circuit of Hadamard test $\cbra{\ket{0^{\log N}}, \Umat_j^\dagger\Umat_{\bm b}} \Gamma_1^2 T$ times, and let the $j$-th element of $\hat{\bm q}$ be $\hat{q}_{j} := \vabs{\bm a_j}\vabs{\bm b} \hat{u}_j $ for any $j\in[M]$\;
Calculate $\hat{\bm x} := \hat{\Wmat}^{-1}\hat{\bm q}$ with classical algorithm\;
\Return $\hat{\bm x}$\;
\caption{Hybrid quantum-classical algorithm for over-determined Linear System.}
\end{algorithm}
\end{minipage}
\end{center}
\end{figure*}
The main quantum part of this algorithm is the Hadamard test, introduced and discussed further in Sec. \ref{subsec:HadamardTest}.
The running time of this algorithm equals $\Ord{M^2 \Gamma_1^2T} + \Ord{M^3}$.
The following theorem shows that Algorithm \ref{Alg:overdetermined} indeed gives a good approximation for the optimal solution $\Amat^{-1}\bm b$.
\begin{theorem}
For the linear system of Problem \ref{pro:overdetermined} in which matrix $\Amat\in \Cbb^{N\times M}$ can be accessed as in Assumption \ref{assume:quantummatrix}, vector $\bm b\in \Cbb^{N}$ can be accessed as in Assumption \ref{assume:vecb}, and $\varepsilon>0$, 
Algorithm \ref{Alg:overdetermined} outputs a classical solution of Problem \ref{pro:overdetermined} with
$\tOrd{\frac{\Gamma_1^2 M^3 \vabs{\Amat^{-1}}^4\kappa(\Amat)^4\vabs{\bm b}^2\pbra{\vabs{\Amat}(\vabs{\bm x^*} + 1)+\varepsilon}^2}{\varepsilon^4}}$
 time, where $\Gamma_1 = \max_{j}\pbra{\vabs{\bm a_j}\vabs{\bm b}, \vabs{\bm a_j}^2}, \bm x^* =\Amat^{-1}\bm b$, 
and high success probability.
\label{thm:SecOverdetermined}
\end{theorem}

\begin{proof}
As before, $\Amat = \sum_{j = 1}^{M}\|\bm a_j\| \ket{a_j}\bra{j}$, where $\ket{a_j}\in \Cbb^{N}$. By Assumption \ref{assume:quantummatrix}, there exist efficient unitary matrices $\Umat_{j}\in\Cbb^{N\times N}$ such that $\Umat_{j} \ket{ 0^{\log N}} = \ket{ a_{j}}
$ for any $j\in[M]$, and $\vabs{\bm a_{j}}$ is given classically. By Assumption \ref{assume:vecb}, $\ket{b}$ can be generated efficiently by the quantum circuit and $\vabs{\bm b}$ is given classically.
Since
\be\argmin_{\bm x} \vabs{\Amat \bm x - \bm b}
=\argmin_{\bm x}\vabs{\Amat^{\dagger}\Amat \bm x- \Amat^{\dagger}\bm b},
\ee
the optimized solution $\bm x^{*} = \pbra{\Amat^{\dagger}\Amat}^{-1}\Amat^{\dagger}\bm b$ can be calculated in $\Ord{M^3}$ time with a classical algorithm if $\Amat^{\dagger}\Amat$ and $\Amat^{\dagger}\bm b$ are given in advance.
We have simply
$\Vmat := \Amat^{\dagger}\Amat =\sum_{j,k}\|\bm a_j\|\|\bm a_k\|\braket{a_j|a_k}\ket{j}\bra{k}$
and $ \bm q := \Amat^{\dagger}\bm{b} = \sum_{j = 1}^{M}\|\bm a_j\|\vabs{\bm b} \ket{j}\braket{a_j|b}$. 
The matrix $\Vmat$ and vector $\bm q$ are consistent with the definitions in Problem \ref{pro:QuadProApp} by letting
$\xi_{jk} = \vabs{\bm a_j}\vabs{\bm a_k}$, $v_{jk} = \braket{ a_j| a_k}$, $\nu_{j} = \vabs{\bm a_j}\vabs{\bm b}$, $u_{j} = \braket{ a_j|b}$, and $K = M$. We define the size parameter $\Gamma_1 := \max_{j}\{\vabs{\bm a_j}\vabs{\bm b}, \vabs{\bm a_j}^2\}$.

We obtain the approximate matrix  $\hat \Vmat$ of $ \Vmat$ and approximate vector $\hat{\bm q}$ of $\bm q$ via the Hadamard test. 
Specifically for any $j,k\in [N]$, the matrix element $\hat{V}_{jk} = \xi_{jk} \hat{v}_{jk}$, where $\hat{v}_{jk} \in \Cbb$ and its real part and imaginary part are the expectations of 
independent runs of the Hadamard tests related to $\cbra{\ket{0^{\log N}},\Umat_j^\dagger\Umat_k}$.
The output distributions of the measurements of the Hadamard tests follows Bernoulli distributions, with expectation values
$\pbra{\real\cbra{\braket{ a_j| a_k}} + 1}/2$, $\pbra{\imag\cbra{\braket{ a_j| a_k}} + 1}/2$,  $\pbra{\real\cbra{\braket{a_j|b}} + 1}/2$, and $\pbra{\imag\cbra{\braket{a_j|b}} + 1}/2$. Hence, we obtain  Bernoulli trials with corresponding expectations $\real\cbra{\braket{ a_j| a_k}}$,
$\imag\cbra{\braket{ a_j| a_k}}$, $\real\cbra{\braket{a_j|b}}$, and $ \imag\cbra{\braket{a_j|b}} $.

It remains to show that the error is bounded for the approximate solution $\hat{\bm x} = \hat {\Wmat}^{-1} \hat{\bm q}$ where $\hat{\Wmat} = \hat \Vmat + \lambda \Imat$ is the shifted matrix of $\hat{\Vmat}$.
We employ Lemma \ref{lem:approxCorrect} to obtain
\be\vabs{\hat{\bm x} - \bm x^*}\leq\frac{\varepsilon}{\vabs{\Amat}},
\ee
which fixes the shifting parameter $\lambda$ and the required 
number of single Hadamard measurements $T$, as follows.
We substitute $\eta \to \varepsilon/\vabs{\Amat}$, 
$\Delta \to \Gamma_1$, and $\Vmat\to\Amat^{\dagger}\Amat$
in Lemma \ref{lem:approxCorrect}.
The lemma hence requires that \be\lambda \leq \frac{\varepsilon/\vabs{\Amat}}{2\vabs{\pbra{\Amat^{\dagger}\Amat}^{-1}}^2 \vabs{\Amat^\dagger\bm b}} \ee 
and 
\be
T = \Ord{\frac{\vabs{\Amat}^2}{\varepsilon^2}\frac{M \pbra{
\vabs{\bm x^*} + 1 + \varepsilon/\vabs{\Amat}
}^2 }{\lambda^2}}.
\ee
Let \be\lambda = \frac{\varepsilon}{2\vabs{\Amat}^2\vabs{\Amat^{-1}}^4 \vabs{\bm b}}\leq \frac{\varepsilon/\vabs{\Amat}}{2\vabs{\pbra{\Amat^{\dagger}\Amat}^{-1}}^2 \vabs{\Amat^\dagger\bm b}},\ee
then $T$ can be bounded as
$ \Ord{\frac{M\vabs{\Amat^{-1}}^4\kappa^4\vabs{\bm b}^2\pbra{\vabs{\Amat}(\vabs{\bm x^*} + 1)+\varepsilon}^2 }{\varepsilon^4}}.$
Each Hadamard test for $\hat{V}_{jk}$ is repeated $\Gamma_1^2  T$ times, hence $\hat{\Vmat}$ is obtained by $M^2 \Gamma_1^2  T$ independent measurements.
Similarly, we obtain $\hat{\bm q}$ with at most the same number of measurements.
Algorithm \ref{Alg:overdetermined} uses these settings to achieve correctness.

The final step is to show the correctness criterion in Problem \ref{pro:overdetermined}. We can project $\bm b$ into the eigenvector space with non-zero eigenvalues of $\Vmat$ and null space of $\Vmat$, \emph{i.e.}, $\bm b = \Vmat^{-1}\Vmat \bm b + P_{N(\Vmat)}\bm b$, and $\bm b - \Amat \bm x^* = P_{N(\Vmat)}\bm b$. Hence,
\be
&&\vabs{\Amat \hat{\bm x} - {\bm b}}
- \vabs{\Amat{\bm x^{*}} - \bm b}\\
&=&\vabs{\Amat\hat{\bm x} - \Amat\bm x^{*} 
- P_{N(\Vmat)}\bm b} -\vabs{P_{N(\Vmat)}\bm b}\\
& \leq &\vabs{\Amat}\vabs{\hat{\bm x} - \bm x^{*}} + \vabs{P_{N(\Vmat)}\bm b} - \vabs{P_{N(\Vmat)}\bm b}\\
&=& \vabs{\Amat}\vabs{\hat{\bm x} - \bm x^{*}} \\
&\leq& \varepsilon.
\ee
\end{proof}

\subsection{Hybrid algorithm for under-determined linear systems\label{subsec:underdetermined}}

With slightly more effort than the previous algorithm we give an algorithm to solve Problem \ref{pro:underdetermined}.
The time complexity is
$\Ord{ M^3\poly(\Gamma_2,\kappa,\vabs{\bm b},\log N)/\varepsilon^{4}}$, where $\sqrt{\Gamma_2}$ is the maximum two norm of the columns of matrix $\Amat$. The challenge here is that  we do not want to  classically write down the high-dimensional vector $\bm y$. To overcome this restriction, we represent solution $\bm y$ as a linear combination of several known vectors, and optimize the coefficients of these vectors with a similar method as Algorithm \ref{Alg:overdetermined}.
The following simple lemma states that the solutions of Problem \ref{pro:underdetermined} can be restricted to the space spanned by $\cbra{\ket{a_1}, \cdots, \ket{a_M}}$, where $\ket{a_j}$ is in proportion to the $j$-th column of matrix $\Amat$.
\begin{lemma}
Given matrix 
$\Amat = \sum_{1\leq j \leq M} \vabs{\bm a_{j}} \ket{ a_j}\bra{j}$,
where $\ket{a_j}\in \Cbb^N$ for any $j\in[M]$, and vector $\bm c \in \Cbb^M$,
then there exists a solution $\bm y^\ast \in \Cbb^N$ in the space spanned by $\cbra{\ket{a_1}, \cdots, \ket{a_M}}$ such that $\bm y^\ast = \argmin_{\bm y}\vabs{\Amat^\dagger  \bm y - \bm c}$.
\label{lem:decompX}
\end{lemma}
\begin{proof}
Let $P$ be the space spanned by $\{\ket{a_1}, \cdots, \ket{a_M}\}$, $P_{\perp}$ be the orthogonal space of $P$.
We can express any solution $\bm y$ as
$\bm y = c_1\bm y^{P} + c_2\bm y^{P_{\perp}}$, for some constants $c_1$ and $c_2$,  where $\bm y^{P}, \bm y^{P_{\perp}}$ are the projections of $\bm y$ into spaces $P$ and $P_{\perp}$ respectively. 
Since $\Amat^\dagger \bm y^{P_{\perp}} =  \sum_{1\leq j \leq M} \vabs{\bm a_{j}} \ket{ j}\bra{ a_{j}}\bm y^{P_{\perp}} = \bm 0$, $\Amat^\dagger \bm y = \Amat^\dagger \bm y^{P}$. Thus we can restrict the solution $\bm y^\ast$ into space $P$.
\end{proof}

In the algorithm below, instead of outputting the $N$ dimensional vector $\bm y$, we output the coefficients of the (non-orthogonal) basis $\cbra{\ket{a_1}, \cdots, \ket{a_M}}$ to construct $\bm y$. The other difference to Algorithm \ref{Alg:overdetermined} is that we have to make the linear system consistent. Consistency is achieved by considering $\hat{\Vmat}^2$ and $\hat{\Vmat} \bm c$ instead of $\hat{\Vmat}$ and $\bm c$.

The following theorem states that Algorithm \ref{Alg:underdetermined} indeed gives a good approximation $\hat{\bm y} := \sum_{1\leq j\leq M}\hat{\alpha}_j\vabs{\bm{a_j}} \ket{ a_{j}}$ for the optimal solution $(\Amat^\dagger)^{-1}\bm c$.
\begin{theorem}
For the linear system of Problem \ref{pro:underdetermined} in which matrix $\Amat^\dagger = \sum_{j = 1}^{M} \vabs{\bm a_j}\ket{j}\bra{a_j}\in \Cbb^{M\times N}$ can be accessed as in Assumption \ref{assume:quantummatrix}, vector $\bm c\in \Cbb^{M}$, and $\varepsilon >0$,
Algorithm \ref{Alg:underdetermined} outputs a vector $\bm s=\pbra{s_1,\cdots, s_M}$, such that $\sum_{j = 1}^M s_j \ket{a_j}$ is a solution of Problem \ref{pro:underdetermined} with high probability, with run time
$\tOrd{\frac{\Gamma_2^2 M^3 \kappa(\Amat)^{12}\vabs{\Amat^{-1}}^4\vabs{\bm c}^2\pbra{\vabs{\bm \Amat}^2\vabs{\bm \alpha^\ast} + \varepsilon + \vabs{\bm c}}^2}{\varepsilon^4}+\frac{\Gamma_2^2 M^3}{\vabs{\Amat}^4}
},$
where $\Gamma_2 :=  \max_{j} \vabs{\bm a_j}^2$.
\label{thm:SecUnderdetermined}
\end{theorem}

\begin{proof}
By Lemma \ref{lem:decompX}, there exits a solution in the space spanned by $\cbra{\ket{a_1}, \cdots, \ket{a_M}}$. 
Let $\bm y := \sum_{k = 1}^M \alpha_k\vabs{\bm{a_k}} \ket{ a_k}$ for some coefficients $\alpha_1, \cdots, \alpha_M \in \Cbb$, which can be written as $ \bm y = \Amat \bm \alpha$ with $\bm \alpha := (\alpha_1,\cdots, \alpha_M)^T$. With $\Vmat := \Amat^{\dagger}\Amat$, we have
 $\Amat^\dagger \bm y = \Vmat \bm \alpha$.
Note that $\Vmat$ matches with the definition of $\Vmat$ in Problem \ref{pro:QuadProApp}.
Therefore, a possible avenue to give the optimized $\bm y$ for Problem \ref{pro:underdetermined} is to find $\bm \alpha$ such that
\begin{align}
    \argmin_{\bm \alpha} 
   \vabs{ \Vmat\bm \alpha -\bm c}.
\label{eq:optPro}
\end{align}  
The optimized solution of Equation \eqref{eq:optPro} equals $\Vmat^{-1}\bm c$, which is also the solution of $\argmin_{\bm \alpha} 
   \vabs{ \Vmat^2\bm \alpha -\Vmat\bm c}$. Hence we would like to optimize
   \begin{align}
       \argmin_{\bm \alpha} 
   \vabs{ \Vmat^2\bm \alpha -\Vmat\bm c}.
   \label{eq:optProConsistent}
   \end{align}
Note that the system $\Vmat^2\bm \alpha = \Vmat\bm c$ is consistent since $\Vmat^2\Vmat^{-2} \Vmat\bm c =\Vmat \bm c$.  Let the shifted matrix be $\Qmat = \Vmat^2 + \lambda \Imat$, where $\lambda >0$ will be specified below.

Algorithm \ref{Alg:underdetermined} constructs a matrix $\hat{\Vmat}$ by Hadamard tests and multiplying by the norms similar to the construction of Algorithm \ref{Alg:overdetermined}.
The algorithm then shifts this matrix as $\hat{\Qmat} := \hat{\Vmat}^2 + \lambda \Imat$. 
The vector $\bm c\in \Cbb^{M}$ can be accessed classically and hence there is no need to approximate it.
The algorithm then solves $\hat{\bm \alpha} =\argmin_{\bm\alpha} \vabs{\hat{\Qmat}\bm\alpha - \hat{\Vmat} \bm c}$ via a classical algorithm to obtain an approximate solution $\hat{\bm \alpha}$ of Equation \eqref{eq:optProConsistent}.
It remains to show that the error is again bounded. 
 
Using the Hadamard tests, 
$\vabs{\hat{\Vmat} - \Vmat} \leq \Ord{\sqrt{M/T}}$ by $\Ord{\Gamma_2^2 M^2 T}$ independent measurements  
with the same analysis as in the proof of Lemma \ref{lem:approxCorrect}. Hence,
\begin{align}
      \vabs{\hat{\Qmat} - \Qmat} 
  &= \vabs{\hat{\Vmat}^2 - \Vmat\hat{\Vmat}
  +\Vmat\hat{\Vmat} - \Vmat^2}\\
 &\leq {\vabs{\hat{\Vmat}}} \vabs{\hat{\Vmat} - \Vmat} + \vabs{\Vmat} \vabs{\hat{\Vmat} - \Vmat}\\
 &\leq {\vabs{\hat{\Vmat} - \Vmat}} \vabs{\hat{\Vmat} - \Vmat} +  2\vabs{\Vmat}\vabs{\hat{\Vmat} - \Vmat}\\
&\leq c\pbra{\sqrt{M/T}  +  \vabs{\Vmat}}\sqrt{M/T}\\
&\leq c'\vabs{\Vmat}\sqrt{M/T},
\end{align}
for suitable constants $c, c',$ and $T\geq M/\vabs{\Vmat}^2$. We proceed analogously to the proof of Lemma \ref{lem:approxCorrect}, where we replace the vectors $\bm q$ and $\hat{\bm q}$ therein with the right-hand sides of the present problem $\Vmat \bm c$ and $\hat{\Vmat}\bm c$, respectively, and set $\eta \to \varepsilon/\vabs{\Vmat}$.
Therefore,  we have
$\vabs{\hat{\bm \alpha} - \bm \alpha^*}\leq \varepsilon/\vabs{\Vmat}$, where $\bm \alpha^{*}$ is the optimized solution of Equation \eqref{eq:optProConsistent},
when the perturbation
\be
\lambda =\frac{\varepsilon}{2\vabs{\Amat^{-1}}^8\vabs{\Amat}^4\vabs{\bm c}} \leq \frac{\eta}{2\vabs{\Vmat^{-2}}^2\vabs{\bm q}},
\ee
and
\begin{widetext}
\be
T&=&\frac{c_1 M \pbra{\vabs{\Vmat}\pbra{\vabs{\bm \alpha^\ast} + \eta}+ \vabs{\bm c}}^2}{\lambda^2 \eta^2} + \frac{M}{\vabs{\Vmat}^2}\\
&=&\frac{c_2 M\pbra{\vabs{\Vmat}\vabs{\bm \alpha^\ast} + \varepsilon + \vabs{\bm c}}^2}{\varepsilon^2/\vabs{\Amat}^4}\times \frac{4\vabs{\Amat^{-1}}^{16}\vabs{\Amat}^8\vabs{\bm c}^2}{\varepsilon^2}+ \frac{M}{\vabs{\Amat}^4}\\
&=& \frac{c_3 M \kappa(\Amat)^{12}\vabs{\Amat^{-1}}^4\vabs{\bm c}^2\pbra{\vabs{\bm \Amat}^2\vabs{\bm \alpha^\ast} + \varepsilon + \vabs{\bm c}}^2}{\varepsilon^4}+ \frac{M}{\vabs{\Amat}^4},
\ee
\end{widetext}
for some large constants $c_1,c_2,c_3$, where the first term comes from Equation \eqref{eq:PartsolBound} in Lemma \ref{lem:approxCorrect}.
With similar analysis to the proof of Theorem \ref{thm:SecOverdetermined}, we have $\vabs{\Vmat \hat{\bm \alpha} - \bm c} - \min_{\bm \alpha}\vabs{\Vmat \bm \alpha - \bm c}\leq \varepsilon$.
Hence $\hat{\bm y} = \sum_{i = 1}^M \hat{\alpha}_i \vabs{\bm a_i} \ket{a_i}$ is a solution of Problem \ref{pro:underdetermined}.
\end{proof}

\begin{figure*}
\begin{center}
\begin{algorithm}[H]
\label{Alg:underdetermined}
\SetKwInOut{Input}{Input}
\SetKwInOut{Output}{Output}
\Input{Matrix $\Amat^\dagger=\sum_{1\leq j\leq M}\|\bm a_j\|\ket{j}\bra{a_j} \in \Cbb^{M\times N}$ that satisfies Assumption \ref{assume:quantummatrix}, vector $\bm c\in \Cbb^{M}$ and $\varepsilon>0$.}
\Output{Vector $\bm s\in \Cbb^M$, where $\hat{\bm y} = \sum_i s_i \ket{a_i}$ is a solution of Problem \ref{pro:underdetermined}.}
Define matrix $\hat{\Vmat}\in \Cbb^{M\times M}$\;
Let $\Gamma_2 :=  \max_{j} \vabs{\bm a_j}^2$, $\lambda :=  \frac{\varepsilon}{2\vabs{\Amat^{-1}}^8\vabs{\Amat}^{4}\vabs{\bm b}}$ and $T := \Ord{\frac{M \kappa(\Amat)^{12}\vabs{\Amat^{-1}}^4\vabs{\bm c}^2\pbra{\vabs{\bm \Amat}^2\vabs{\bm \alpha^\ast} + \varepsilon + \vabs{\bm c}}^2}{\varepsilon^4}}+ \frac{M}{\vabs{\Amat}^4}$\;
Let $\hat{v}_{j,k} $ be the approximate value of $\braket{a_j|a_k}$ by repeatedly measuring quantum circuit of Hadamard test $\cbra{\ket{0^{\log N}}, \Umat_j^\dagger \Umat_k}$ $\Gamma_2^2 T$ times\;
 $\hat{V}_{j,k} := \vabs{\bm a_j} \vabs{\bm a_k} \hat{v}_{j,k}$ for any $j,k\in[M]$\;
Let matrix $\hat{\Qmat}:=\hat{\Vmat}^2 + \lambda \Imat$\;
Calculate $\hat{\bm \alpha} := \hat{\Qmat}^{-1}\hat{\Vmat}{\bm c}$ with classical algorithm\;
Calculate $\bm s =\pbra{\hat{\alpha}_1\vabs{\bm a_1}, \cdots, \hat{\alpha}_M \vabs{\bm a_m}}$\;
\Return $\bm s$\;
\caption{Hybrid quantum-classical algorithm for under-determined linear system.}
\end{algorithm}
\end{center}
\end{figure*}

Note that the number of measurements $T$ can be simplified to $\Ord{M\kappa(\Amat)^{16}\pbra{\vabs{\bm \alpha^\ast} + 1}^2/\varepsilon^4}$ when $\vabs{\bm c}=\Ord{1}$,
$\vabs{\Amat} = \Theta(1)$ and $\varepsilon = \Ord{1}$.
The same simplifications are used in the comparison Table \ref{tab:ComparisonAlg}.

We now discuss some ways of using the output of Algorithm \ref{Alg:underdetermined} and refer to Appendix \ref{app:ApplicationErrorBound} for more details.
Let $\hat{\bm \alpha}$ satisfy $\vabs{\hat{\bm \alpha} - \bm \alpha}\leq \eta$, and $s_i = \alpha_i \vabs{\bm a_i}$, where $\eta$  depends on how long we run Algorithm \ref{Alg:underdetermined}. Take $\eta$ to be small enough such that it compounds with errors arising from the additional quantum circuits and measurements to a final error $\varepsilon$, see Appendix \ref{app:ApplicationErrorBound}.
There are several applications with the hybrid form $\cbra{s_i, \ket{a_i}:i\in[M]}$ of our output $\hat{\bm y}$. 
For example, it can be applied to evaluate the inner product with another quantum state.
Given proper $\eta$, a vector $\ket{v}\in\Cbb^{N}$, and the $\polylog N$-depth quantum circuit $\Umat_{v}$ such that $\ket{v} = \Umat_{v}\ket{0^{\log N}}$. We can approximate the inner product $\bra{v}\bm y = \sum_{i = 1}^M s_i \braket{v|a_i}$ with Hadamard tests of $\cbra{\ket{0^{\log N}}, \Umat_v^{\dagger} \Umat_j}$ for all of $j\in[M]$, with $\Ord{M\vabs{\Amat}_F^2 \vabs{\bm \alpha}^2/\varepsilon^2}$ independent measurements, and additive error $\varepsilon$ in total. 

The hybrid state can also be applied to 
measure Hermitian operators. 
Given proper $\eta$, $\Ord{\polylog N}$-depth quantum circuits of $\Hmat_k$ for $k\in[K_H]$ and $K_H = \Ord{\polylog N}$. Let the Hermitian operator be given as $\Hmat = \sum_{k =1}^{K_H} \gamma_k \Hmat_k$, where $0<\gamma_k<\Delta_H$ are coefficients.
The expectation $\bm y^{\dagger}\Hmat \bm y$ of $\Hmat$ can be approximated with $\tOrd{\frac{M^2 \Delta_H^2\vabs{\Amat}_F^4\vabs{\bm \alpha}^4}{\varepsilon^2}}$ independent measurements and $\varepsilon$ additive error.

Quantum systems are disturbed by interaction with an environment, leading to loss of coherence which especially affects quantum circuits with high depth. Hence, it is a worthwhile goal to make the above algorithms more practical for near-term quantum devices by reducing and optimizing the circuit depth.
To this end, we consider the circuit optimization of the Hadamard test in the next section.

\section{Optimization of the circuit of the Hadamard test\label{sec:CirHadamardTest}}

This section introduces the techniques for circuit optimization and applies them to the optimization of the depth of the Hadamard test. There are several approaches to construct quantum computers, including ion traps~\cite{figgatt2019parallel,wright2019benchmarking} and superconducting systems ~\cite{arute2019quantum, IBMQ, ye2019propagation}. Both approaches have their advantages and disadvantages. We propose two circuit optimization algorithms which are suitable for these two structures, respectively. The first algorithm is suitable for ion trap quantum computers, in which all of qubits are connected. The second algorithm is suitable for present-day superconducting quantum devices, in which the connection of the qubits can be represented as a lattice.

Sec.~\ref{subsec:HadamardTest} gives an overview of the Hadamard test and some basic concepts about quantum circuits, Sec. \ref{subsec:CompleteOpt} gives an optimization algorithm to optimize the depth of the Hadamard test under a complete graph, in which a CNOT gate can operate on any two qubits. Sec. \ref{subsec:latticeOpt} gives an algorithm to optimize the depth of the Hadamard test under a graph, here the 2D lattice, and generalizes it to any connected graph with a Hamiltonian path. For convenience, we denote $n :=\ceil{\log N}$ and set the number of qubits to be $n$.

\subsection{Some notation and review of Hadamard test\label{subsec:HadamardTest}}

\textbf{Definition of Gates:}
$\CNOT_{ij}$ means CNOT gate with control qubit $i$ and target qubit $j$. An $n$-qubit CNOT circuit is any $n$-qubit quantum circuit consisting only of CNOT gates.
Toffoli$_{ijk}$ means Toffoli gate with control qubits $i, j$ and target qubit $k$. We use $\Rmat_j(\theta)$ to represent the single-qubit rotation gate with parameter $\theta$ on qubit $j$. Here parameter $\theta$ contains the information of angle as well as the axis of rotation with a little abuse of notation. 

\textbf{Depth of a circuit:}
The depth of a circuit is the maximal length of a path from any initial gate to any of final gates in the circuit, where the gate set contains single and two-qubit gates (and not three-qubits gates, for example). A $d$-depth circuit $\Ccal$ is constructed by $d$ one layer circuits: $\Ccal = \Ccal^{(d)}\cdots \Ccal^{(1)}$, where $\Ccal^{(i)}$ is the $i$-th layer of circuit.

\textbf{Quantum system under a graph:}
For an $n$-qubit quantum system, each qubit is represented as a vertex $v \in V$ in the graph $G(V,E)$ with edge set $E$, where $|V|=n$. The quantum gates  can only operate on single qubits or two qubits which are connected in $G$, and we say this $n$-qubit quantum system is under the graph $G$ (or constrained to graph $G$).

\textbf{Permutation $\bm \sigma$ on graph $G(V,E)$:}
a vertex $v_i\in V$ is labeled to $\sigma_i\in [|V|]$, where $\bm \sigma = \pbra{\sigma_1, \cdots, \sigma_{|V|}}$ is a permutation ($\sigma_i \ne \sigma_j$ for any $i\ne j$).

\textbf{Snakelike order:}
For a lattice graph $G(V,E)$, the snakelike order of $G$ is to label the top left corner to $1$, and then proceed to the rightmost in increasing order, and then continue in the same fashion for the next row. In other words, the odd rows are labelled in an increasing order, the even rows are labelled in a decreasing order, while the columns are arranged in the increasing order. Figure \ref{fig:Snakelikelattice} shows the snakelike order on $4\times 4$ lattice.

\begin{figure}
    \centering
    \includegraphics[width=0.2\textwidth]{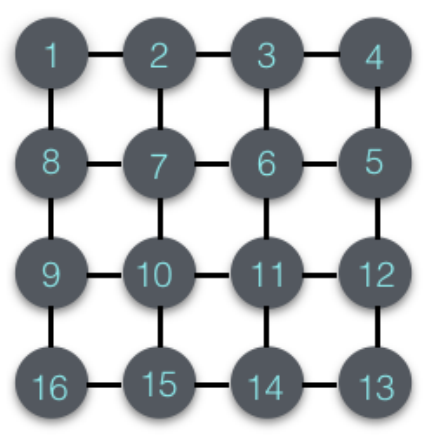}
    \caption{An example for snakelike order on $4\times 4$ lattice.}
    \label{fig:Snakelikelattice}
\end{figure}

\textbf{Hadamard test:} Hadamard test of $\cbra{\ket{0}, \Umat}$
 is to calculate the value $\bra{0}\Umat\ket{0}$ by performing quantum circuits, given $n$-qubit quantum state $\ket{0}$ and unitary $\Umat$.
Figure \ref{fig:CirHadamard} depicts the circuit of the Hadamard test. It is easy to check that when the first qubit is finally measured in the computational basis, the probability to obtain the `0' outcome $\Pr(0) =\pbra{1 + \real{\bra{0}\Umat\ket{0}}}/2$ for $j = 0$, and $\Pr(0) = \pbra{1 + \imag{\bra{0}\Umat\ket{0}}}/2$ for $j = 1$ in Figure \ref{fig:CirHadamard}.
Independently perform this process $\Ord{1/\varepsilon^2}$ times obtains $\bra{0}\Umat\ket{0}$ with additive error $\varepsilon$ and high success probability by Hoeffding's inequality.

However, given a $d$-depth circuit $\Ccal$ of unitary $\Umat$, the naive method to perform the Control-$\Ccal$ operation requires $\Ord{dn}$-depth if we do not restrict the connectivity of any two qubits, and requires $\Ord{dn^2}$-depth for any connected graph by Refs.~\cite{kissinger2020cnot,nash2020quantum,wu2019optimization}, which is costly for large qubit systems, especially when $n \gg d$.  Hence in the next two subsections, we focus on the depth optimization of Control-$\Ccal$ circuit for a given $d$-depth circuit $\Ccal$ of unitary $\Umat$.

\begin{figure}[h]
    \centering
    \[  \Qcircuit @C=1em @R=.7em {
  \lstick{\ket{0}}   &\qw & \gate{H}  &\ctrl{1} & \gate{S^{j}} & \gate{H}  & \meter\\
   \lstick{\ket{0}}  & \ustick{n} \qw &{/}\qw & \gate{\Umat} &\qw & \qw & \qw
    }\]
    \caption{Circuit for Hadamard test, $\Pr(0) =\pbra{1 + \real{\bra{0}\Umat\ket{0}}}/2$ when $j = 0$, and $\Pr(0) = \pbra{1 + \imag{\bra{0}\Umat\ket{0 }}}/2$ when $j = 1$. Here $H$ is Hadamard gate and $S$ is phase gate. }
    \label{fig:CirHadamard}
\end{figure}

\subsection{Depth optimization of Hadamard test under complete graph\label{subsec:CompleteOpt}}

The depth of Hadamard test of $\{\ket{0}, \Umat\}$ is of the same order as the depth of the Control-$\Ccal$ circuit --- the core part of Hadamard test. Recall that $\Ccal$ is the quantum circuit implementation of $\Umat$ with depth $d$. 
Here we introduce an optimized Control-$\Ccal$ circuit in the fully-connected setting. We prove that the optimized depth of Control-$\Ccal$ equals $\Ord{\log s + d\pbra{\log \frac{n}{s} + 1}}$ with $s$ ancillas. The main result is Lemma \ref{lem:sancillas} which immediately leads to 
Theorem \ref{thm:sancillasIntro}.

We also prove that asymptotically $\Omega( \log n)$ quantum depth are required to compute $\bra{0^{n}}\Ccal\ket{0^{n}}$ regardless of the number of ancillas and the connectivity of the qubits, see Theorem \ref{thm:LowerboundHadamardIntro}. Hence,  for $d=\Ord{\log n}$ and $s = \Theta(n)$, the optimized depth given by our algorithm is asymptotically tight with $s$ ancillas. 
The following claim and lemma are precursors to the circuit optimization algorithm of Lemma \ref{lem:sancillas}.
\begin{claim}
Given an $(n-1)$-depth CNOT circuit constructed by $\CNOT_{1,j}$ for $2\leq j \leq n$, there exists an equivalent CNOT circuit with depth $2\ceil{\log n} - 1$ and without ancillas.
\label{claim:CnotOneCTGate}
\end{claim}
This CNOT circuit can be obtained easily by the optimized CNOT circuit $\CNOT_{j,n}$ for $1\leq j \leq n$ in Ref.~\cite{moore2001parallel}. For the completeness of this manuscript, we also give the proof of this claim in Appendix \ref{app:copy}. The optimized depth also matches the lower bound $\Omega(\log n)$~\cite{jiang2020optimal}.

\begin{lemma}
For any $n$-qubit and one-depth circuit $\Ccal$,
the depth of Control-$\Ccal$ can be paralleled to $12\ceil{\log n} +9$ without ancillas, using CNOT and single-qubit gates.
\label{lem:OneLayerLog}
\end{lemma}

We postpone the proof of this lemma to Appendix \ref{app:onelayercontrol}. The idea of the proof is optimizing the control circuit of $\Ccal$ where $\Ccal$ is one layer of single-qubit gates $\Rmat(\theta)$ and one layer of CNOT gates respectively.
By this lemma, given $\Ord{1}$-depth circuit $\Ccal$, $\Ord{\log n}$-depth is enough to construct Control-$\Ccal$, which also matches the lower bound, since for some $\Ord{1}$-depth $\Ccal$, the depth of Control-$\Ccal$ is $\Omega\pbra{\log n}$. The following lemma generalizes the result for one-depth circuits of Lemma \ref{lem:OneLayerLog} to $d$-depth circuits.

\begin{lemma}
For any $n$-qubit and $d$-depth circuit $\Ccal$, the depth of Control-$\Ccal$ can be paralleled to $2\ceil{\log s} + 12d\ceil{\log \frac{n}{s}} +9d $ with $s-1$ ancillas where $s\in[n]$, using CNOT and single-qubit gates.
\label{lem:sancillas}
\end{lemma}

\begin{proof}
Let the $d$-depth circuit be  $\Ccal := \Ccal^{(d)}\Ccal^{(d-1)}\cdots \Ccal^{(1)}$, where $\Ccal^{(j)}$ is one layer of the circuit.
 Let qubits $2,\cdots, s$ be ancillas, and initialized to $\ket{0}$, where $s\in [n]$, as shown in Fig. \ref{fig:ctrlOneCnot}. Let $\ket{\psi_0} = \pbra{\alpha\ket{0} + \beta\ket{1}}\ket{0^{s-1}}\ket{\phi}$ be the initial state, where $\alpha, \beta \in \Cbb$, and $|\alpha|^2 + |\beta|^2 = 1$.
 We generate transformation Control-$\Ccal$ by the following process:
\begin{itemize}
    \item [(1)] Copy step: $\ket{\psi_0}\rightarrow \ket{\psi_1}:= \pbra{\alpha\ket{0^s} + \beta \ket{1^s}} \ket{\phi}$.
    (Via the copy circuit~\cite{moore2001parallel}.)
 \item [(2)] Control of $\Ccal$: $\ket{\psi_1}\rightarrow \ket{\psi_2}:= \alpha \ket{0^s}\ket{\phi} + \beta \ket{1^s}\Ccal \ket{\phi}$. (For each depth $\Ccal^{(j)}$, 
we divide the gates into $s$ sets such that each set $S_i$ occupies at most $\ceil{n/s}$ qubits. Use the $i$-th qubit to control the gates in the set $S_i$ for $i\in [s]$.)
 \item [(3)] Uncompute step: $\ket{\psi_{2}} \rightarrow \alpha\ket{0}\ket{0^{s-1}}\ket{\phi} + \beta \ket{1}\ket{0^{s-1}}\Ccal\ket{\phi}$. (Via the inverse operations of the Copy step.)
\end{itemize}

By Moore \emph{et al.}~\cite{moore2001parallel}, the depth of the Steps (1) and (3) are both $\ceil{\log s}$.
By Lemma \ref{lem:OneLayerLog}, we can optimize the control of the $n/s$ targets to $12\ceil{\log \frac{n}{s}} + 9$-depth for each $\Ccal^{(j)}$ in Step (2). Thus the total depth of this process equals $2\ceil{\log s} + d\pbra{12\ceil{\log \frac{n}{s}} + 9}$ in the worst case.
\end{proof}

\begin{figure*}
\begin{center}
\hspace{1em}
\begin{tabular}{c}
\vspace{-.6em}\\
\begin{pgfpicture}{0em}{0em}{0em}{0em}
\color{cyan}
\pgfsetdash{{0.2cm}{0.2cm}{0.2cm}{0.2cm}}{0cm}
\pgfxyline(1,0.1)(3.5,0.1)
\pgfxyline(3.5,0.1)(3.5,-2.0)
\pgfxyline(3.5,-2.0)(1,-2.0)
\pgfxyline(1,-2.0)(1,0.1)
\color{cyan}
\pgfsetdash{{0.2cm}{0.2cm}{0.2cm}{0.2cm}}{0cm}
\pgfxyline(8.5,0.1)(11,0.1)
\pgfxyline(11,0.1)(11,-2.0)
\pgfxyline(11,-2.0)(8.5,-2.0)
\pgfxyline(8.5,-2.0)(8.5,0.1)
\color{darkgray}
\pgfrect[fill]{\pgfpoint{10.8em}{0.2em}}{\pgfpoint{5em}{-10em}}
\color{black}
\pgfsetdash{{0.1cm}{0.2cm}{0.2cm}{0.1cm}}{0cm}
\pgfxyline(4.1,0.1)(6.1,0.1)
\pgfxyline(6.1,0.1)(6.1,-3.8)
\pgfxyline(6.1,-3.8)(4.1,-3.8)
\pgfxyline(4.1,-3.8)(4.1,0.1)
\color{red}
\pgfsetdash{{0.1cm}{0.2cm}{0.2cm}{0.1cm}}{0cm}
\pgfxyline(4,-2)(8,-2)
\pgfxyline(8,-2)(8,-5)
\pgfxyline(8,-5)(4,-5)
\pgfxyline(4,-5)(4,-2)
\end{pgfpicture}
\Qcircuit @C=1em @R=0.7em{
&\qw
&\ctrl{2} & \ctrl{1} & \qw 
& \qw 
& \ctrl{4}  & \ctrl{5} & \qw
& \qw 
 &\qw 
& \ctrl{1} & \ctrl{2} & \qw \\
& \qw
&\qw & \targ & \qw 
& \qw 
& \qw & \qw & \qw
& \qw
& \qw 
&\targ &\qw & \qw \\
& \qw
&\targ & \ctrl{1} & \cds{1}{\cdots}
& \qw  
& \qw  & \qw & \ctrl{5}
& \qw 
& \cds{1}{\cdots}
 & \ctrl{1} & \targ & \qw \\
& \qw
&\qw & \targ &\qw
&\qw 
& \qw &\qw & \qw
& \qw 
& \qw
 &\targ & \qw & \qw \\
& \qw
&\qw & \qw  & \qw 
&\qw 
&\gate{\Rmat(\theta_1)} & \qw & \qw
& \qw 
& \qw 
&\qw & \qw & \qw \\
 & \qw
&\qw & \qw  & \qw 
&\qw 
&\qw & \ctrl{1} & \qw 
& \qw 
& \qw 
&\qw & \qw & \qw \\
& \qw
&\qw & \qw  & \qw 
&\qw 
&\qw & \targ & \qw 
& \qw 
& \qw 
&\qw & \qw & \qw \\
&\qw 
& \qw  & \qw  &\qw
&\qw 
&\qw & \qw & \gate{\Rmat(\theta_2)}
& \qw 
& \qw 
&\qw & \qw & \qw 
\inputgroupv{5}{8}{.8em}{.8em}{\ket{0}^{\otimes n}}
\inputgroupv{1}{4}{.8em}{.8em}{\ket{0}^{\otimes s}}
}
\vspace{1.2em}\hspace{1.2em}
\\
\end{tabular}
\caption{Circuit for $\Control$-$\Ccal$ with $s$ ancillas, in which the circuit of the left-top dashed box perform the copy operation, and the right-top box perform the uncompute operation, the red dashed box is the circuit $\Ccal$, and the box in gray is to target gates with one control qubit.}
\label{fig:ctrlOneCnot}
\end{center}
\end{figure*}
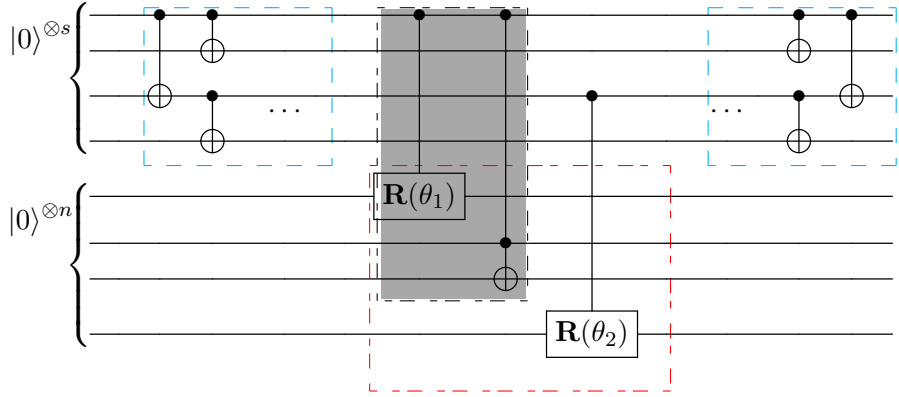

Lemma \ref{lem:sancillas} indicates that the depth of the Hadamard test of $\{\ket{0},\Umat\}$ can be optimized to $\Ord{\log s + d(\log n/s + 1)}$ with $(s + n)$ qubits, where $s \in [n]$ and $d$ is the depth of circuit representing $\Umat$. Notice that when $s = n$, the depth can be optimized to $\Ord{d + \log n}$, which matches the lower bound for some particular circuits. 

However, in real superconducting quantum device, not any two qubits are connected~\cite{arute2019quantum, boixo2018characterizing,IBMQ, gong2019genuine}, the connectivity of physical qubits are constrained to a graph. In next subsection, we will give the algorithm to map the circuit of Hadamard test to physical device with lattice structure.

\subsection{Depth optimization of Hadamard test under lattice graph\label{subsec:latticeOpt}}

Here we consider a special graph --- a two-dimensional lattice. Such a graph commonly occurs in physical quantum computing implementations~\cite{arute2019quantum,boixo2018characterizing,IBMQ}.  Lemma \ref{lem:2DControlU} shows our optimization result.
Before discussing it, we give a claim and several related lemmas to serve for it.

\begin{claim}
On a lattice of size $l_1 \times l_2$ where one can only perform the SWAP operation between the neighbors of the lattice, there exists an $\Ord{l_1+l_2}$ parallel time algorithm, which can generate any permutation ${\bm \sigma}=\pbra{\sigma_1, \cdots, \sigma_n}$ from the snakelike order.
\label{claim:perm2Dim}
\end{claim}
The permutation problem in Claim \ref{claim:perm2Dim} is exactly the inverse process of parallel sorting problem on $l_1\times l_2$ 2D lattice in Sec 9.4 of Ref.~\cite{parhami2006introduction}. It can be solved in $\Ord{l_1+l_2}$ parallel time using a recursive  sorting algorithm. Hence, the permutation problem in Claim \ref{claim:perm2Dim} needs $\Ord{l_1 + l_2}$-depth of the SWAP operations.

The following lemma states how to map one layer of CNOT gates to a physical device where the connectivity of qubits is constrained to an $l_1\times l_2$ lattice.
\begin{lemma}
Suppose we are given an $n$-qubit and one-depth CNOT circuit $\Ccal$ and we have to re-design it such that the qubits are constrained according to an $l_1 \times l_2$ lattice, where $l_1l_2 = n$. There exists an equivalent $\Ord{l_1 + l_2}$-depth CNOT circuit under the above lattice. Furthermore, this complexity bound is asymptotically tight. 
\label{lem:OneLayerCNOTTopo}
\end{lemma}

\begin{proof}
Observe that we can represent a one-depth and $n$-qubit CNOT circuit as
\be
\Ccal = \bigotimes_{i = 1}^{r} \CNOT_{s_i,t_i},
\ee
where $s_i, t_i\in[n], r\leq n/2$, and any qubit $q\in[n]$ appears at most once in the multiset $\{s_1,s_2,\ldots, s_r\} \cup \{t_1,t_2,\ldots, t_r\}$. 

We first label the qubits in the $l_1\times l_2$ lattice in the snakelike order. Next, permute all of qubits to order $\sigma$ such that CNOT gates arise only in neighbors of the new order $\sigma$.  Next, apply CNOT gates in one layer since all of gates are operated on neighbors of qubits. 
In the end, restore the original order of qubits by performing the inverse process of the first step --- permuting qubits. A simple example of this step is described in Fig~\ref{fig:toyeg}.

By Claim \ref{claim:perm2Dim}, we can generate any permutations with $\Ord{l_1+l_2}$-depth of SWAP gates. Since one SWAP gate can be constructed by 3 CNOT gates, any permutations can also be generated by $\Ord{l_1+l_2}$-depth of CNOT gates. Meanwhile, the inverse process of the permutation has the same depth, thus the total depth to map the one layer of CNOT gates to $l_1 \times l_2$ lattice is $\Ord{l_1+l_2}$.

For the lower bound, observe that it needs at least depth $l_1 + l_2$ to map a CNOT gate operating on two qubits corresponding to vertices with distance $l_1+l_2$ (such as left-top and right-bottom vertices) to the CNOT circuit under this $l_1 \times l_2$-lattice.
\end{proof}

\begin{figure}
  \centering
  \begin{tabular}[b]{c}
    \includegraphics[width=.3\linewidth]{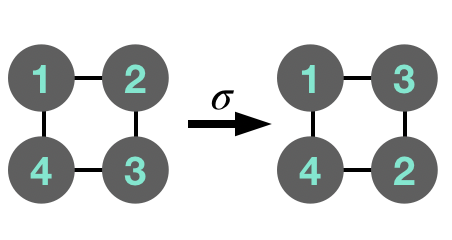} \\
    \small (a)
  \end{tabular} \qquad
  \begin{tabular}[b]{c}
     \Qcircuit @C=1em @R=.7em {
  \lstick{1}  & \qw & \ctrl{2} & \qw & \qw\\
  \lstick{2} & \qw &  \qw & \ctrl{2}  & \qw \\
  \lstick{3} & \qw &  \targ   & \qw & \qw  \\
  \lstick{4} & \qw &  \qw & \targ   & \qw
  }
     \\
    \small (b)
  \end{tabular}
  \qquad
  \begin{tabular}[b]{c}
  \Qcircuit @C=1em @R=.7em {
  \lstick{1}  & \qw & \ctrl{1} & \qw & \qw\\
  \lstick{3} & \qw &  \targ & \qw  & \qw \\
  \lstick{2} & \qw &  \qw   & \ctrl{1} & \qw  \\
  \lstick{4} & \qw &  \qw & \targ   & \qw
  }
       \\
       \small (c)
  \end{tabular}
  \caption{(a) $2 \times 2$ lattice in snake order and a permutation $\sigma = (1,3,2,4)$. (b) A depth-1, 4 qubit CNOT circuit. (c) The same circuit as in (b) with qubits 2 and 3 swapped.}
  \label{fig:toyeg}
\end{figure}

The following lemma states how to map a particular  CNOT circuit with one common control qubit and $k<n$ targets to a lattice. The size of this particular kind of circuits can be optimized to $\Ord{n}$ with the optimization technique in Ref.~\cite{nash2020quantum}.
One can easily parallelize the optimization process of Ref. \cite{nash2020quantum} to $\Ord{l_1 + l_2}$-depth on $l_1 \times l_2$ lattice. For completeness, we give a simpler algorithm in Appendix \ref{app:CtrlCNOTLattice} which has the same parallel depth.

\begin{lemma}
Given a CNOT circuit constructed by $\CNOT_{1,i_1},\cdots, \CNOT_{1,i_k}$, where $1< i_1 <\cdots < i_k \leq n, k<n$, and an $l_1\times l_2$ lattice as the constrained graph of the quantum system, there exists an equivalent $\Ord{l_1+l_2}$-depth CNOT circuit under the above lattice.
\label{lem:CNOTMap2lattice}
\end{lemma}
The proof of Lemma \ref{lem:CNOTMap2lattice} is in Appendix \ref{app:CtrlCNOTLattice}.
Lemma \ref{lem:CNOTMap2lattice} also matches the lower bound, since it needs $\Omega\pbra{l_1 + l_2}$-depth to implement the CNOT gate on two qubit with largest distance $l_1 + l_2$. 

\begin{lemma}
Suppose we are given an $n$-qubit and $d$-depth circuit $\Ccal$ and we have to re-design it such that the qubits are constrained according to an $l_1 \times l_2$ lattice, where $l_1l_2 = n$. There exists an $\Ord{d\pbra{l_1 + l_2}}$-depth circuit for Control-$\Ccal$ under the above lattice.
\label{lem:2DControlU}
\end{lemma}

\begin{proof}
Suppose $\Ccal = \Ccal^{(d)}\cdots \Ccal^{(1)}$, where $\Ccal^{(i)}$ is the $i$-th layer of $\Ccal$.
We decompose each Control-$\Ccal^{(i)}$ for $i\in [d]$ into CNOT gates and single-qubit gates via the method of Lemma \ref{lem:OneLayerLog}. 
Since here a CNOT (Toffoli) gate can operate on two (three) qubits only when they are neighbors on the lattice, we need to map the following two cases of circuits to the lattice:
\begin{itemize}
    \item One layer of Toffoli gates;
    \item The CNOT circuit which is constructed by $\CNOT_{1,i_1},\cdots, \CNOT_{1,i_k}$, where $1< i_1 <\cdots < i_k \leq n$.
\end{itemize}

For one layer of Toffoli gates, we permute the qubits to let all of Toffoli gates operate on neighbors of qubits.
For one Toffoli gate, we firstly decompose it into 6 layers of CNOT gates \cite{shende2008cnot} joint with single-qubit gates in alternatively. The decomposed CNOT gate can be implemented by at most 4 CNOT gates on the lattice 
since the distance\footnote{Distance: The number of edges for the minimum path of these two vertices in the lattice.} between these two qubits is at most 2.
Hence after permutation, we can implement one layer of Toffoli gates with constant depth of CNOT gates joint with constant depth of single-qubit gates (12 layers of CNOT gates), and since permutation operation and its inverse process both can be implemented in $\Ord{l_1+l_2}$-depth by Lemma \ref{claim:perm2Dim}, the implementation of one layer of Toffoli gates on $l_1 \times l_2$ lattice needs $\Ord{l_1+l_2}$-depth in total.

For CNOT circuit which is constructed by $\CNOT_{1,i_1},\cdots, \CNOT_{1,i_k}$, where $1< i_1 <\cdots < i_k \leq n$, it can be implemented on $l_1\times l_2$ lattice with depth $\Ord{l_1+l_2}$ by Lemma \ref{lem:CNOTMap2lattice}.

Thus mapping Control-$\Ccal$ for depth $d$ circuit $\Ccal$ to $l_1\times l_2$ lattice in which $l_1 l_2 = n$ needs $\Ord{d\pbra{l_1 + l_2}}$-depth. 
\end{proof}

\begin{corollary}
Suppose we are given an $n$-qubit and $d$-depth circuit $\Ccal$ and we have to re-design it such that the qubits are constrained according to a graph $G$ which has a Hamiltonian path. There exists an $\Ord{dn}$-depth circuit for Control-$\Ccal$ under the above graph.
 \label{coro:TopoAnyGCtrlU}
\end{corollary}

The corollary can be obtained directly by Lemma \ref{lem:2DControlU} by letting $l_1=n, l_2=1$.
By Lemma \ref{lem:2DControlU} (Corollary \ref{coro:TopoAnyGCtrlU}), we can immediately obtain Theorem \ref{thm:2DControlUIntro} (Corollary \ref{coro:TopoAnyGCtrlUIntro}).

In the following, we prove that there exists some unitary $\Umat$, the Hadamard test of $\cbra{\ket{0}, \Umat}$ needs at least $\Omega\pbra{\max\cbra{\log n, D}}$-depth of quantum circuits under a graph with diameter $D$.
\begin{proof}[Proof of Theorem \ref{thm:LowerboundHadamardIntro}]
For the specific unitary $\Umat = \otimes_{j = 1}^{n}\Rmat_j$, where $\Rmat_j = e^{i\theta_{j}}\Imat$, then $\bra{ 0^{n}}\Umat\ket{ 0^{n}} = \bra{ 0^{n}}\otimes_{j = 1}^{n}e^{i\theta_j}\Imat\ket{ 0^{n}} = e^{i\sum_{j=1}^{n}\theta_{j}}$. Since there needs $\Omega(\log n)$-depth to compute $\theta_1 +\cdots + \theta_n$ for complete graph in which any two qubits are connected. And for any constrained graph with diameter $D$, it needs at least $\Omega\pbra{D}$-depth to generate $e^{i\sum_{j=1}^{n}\theta_{j}}$, since there exists two qubit $j,k$, we need at least $D$-depth to compute the addition of $\theta_j + \theta_k$.
Thus we need at least $\Omega(\max\cbra{\log n, D})$-depth to generate $\bra{ 0^{n} }\Umat\ket{0^{n}}$.
\end{proof}

\section{Discussion\label{sec:discussion}}

In this work we have investigated the solving of skewed linear systems with near-term hybrid algorithms.
We have assumed quantum access to the normalized columns (or rows, respectively) of the corresponding matrix with efficiently implementable quantum circuits and classically given norms.
We have proposed two algorithms for two different kinds of systems --- over-determined and under-determined. The time complexity of our algorithms are polynomial in the smaller dimension of matrix and logarithmic in the larger dimension. Hence we can solve certain skewed linear systems efficiently in the larger dimension. With similar techniques, we can also solve low-rank factorized linear systems efficiently. See Table \ref{tab:ComparisonAlg} for a comparison to other linear systems solvers with different data input models. 

Our hybrid algorithms require the Hadamard test on the quantum computer as a subroutine. Since quantum systems tend to decohere in a short amount of time, we have also considered the optimization of the depth of Hadamard test $\cbra{\ket{0}, \Umat}$.
The first case we have considered is when two-qubit gates can operate on any two qubits, the second case is when there is a lattice underlying the connectivity of the qubits, similar to recent near-term quantum devices.

We leave several open problems based on this work.
\begin{itemize}
    \item Does there exist any hybrid near-term linear system algorithm with  time complexity that is linear or sub-linear in $1/\varepsilon$?
    \item More generally, can the time complexity of this work be improved in some or all of the parameters?
    \item Are our circuit optimization algorithms tight beyond one-depth circuits? Given a $d$-depth circuit $\Ccal$, is there any better lower bound than Theorem \ref{thm:LowerboundHadamardIntro} for the depth of quantum circuit to compute 
    $\bra{0} \Ccal \ket{0}$?
\end{itemize}

\begin{table*}
\begin{center}
\begin{tabular}{c|c|c}
\hline
Reference & Time complexity & Input model\\
\hline
SV09 \cite{strohmer2009randomized}& $\Ord{M \vabs{\Amat}_F^2\vabs{\Amat^{-1}}^2\log (\varepsilon_0/\varepsilon)}$& Cons., full col. rank, sampling access\\
CLW18 \cite{chia2018quantum}& $\tOrd{\vabs{\Amat}_F^{2\omega}\kappa^{8\omega + 4}\vabs{\bm x^*}^{2\omega}/{\varepsilon^{2\omega}}}$& Length-square sampling \cite{tang2019quantum}\\
GLT18 \cite{gilyen2018quantum}& $\tOrd{\vabs{\Amat}_F^6\kappa^{16}R^6\vabs{\bm x^*}^6/{\varepsilon^6}}$& Length-square sampling \cite{tang2019quantum}\\
Problem \ref{pro:overdetermined} (Algorithm \ref{Alg:overdetermined}) &$\tOrd{\Gamma^2 M^3\kappa^8\pbra{\vabs{\bm x^*}^2 + 1}/\varepsilon^4}$ & Quantum circuits access\\
Problem \ref{pro:underdetermined} (Algorithm \ref{Alg:underdetermined})& $\tOrd{\Gamma^2 M^3\kappa^{16}\pbra{\vabs{\bm \alpha^*}^2+1}/\varepsilon^4}$ 
&Quantum circuits access\\
Problem \ref{pro:lowrank}  &$\tOrd{\Gamma'^2R^3\pbra{
\kappa(\Amat_1)\kappa(\Amat_2)}^6\kappa^2/\varepsilon^2}$&Quantum circuits access\\
\hline
\end{tabular}
\end{center}
\caption{The comparison between our algorithms and the existing algorithms for solving linear system $\Amat \bm x = \bm b$ for $\Amat\in \Cbb^{N\times M}$ with $\vabs{\Amat\hat{\bm x} -\bm b } - \vabs{\Amat\bm x^* - \bm b}\leq \varepsilon$, where $\hat{\bm x}$ and $\bm x^*$ are the approximate solution and optimal solution respectively. 
Here, $\varepsilon_0\leq \|\bm x^*\|, \omega <2.373$ is the matrix multiplication exponent, $\kappa$ is the condition number of $\Amat$, and $R$ is the rank of $\Amat$.
In addition, $\bm a_i$ is the $i$-th column of matrix $\Amat$,  and $\Gamma=\max_i\cbra{\vabs{\bm a_i},\vabs{\bm a_i}^2}$. 
For Problem \ref{pro:underdetermined}, note that we are solving  $\Amat^\dagger \bm y = \bm c$ with 
$\bm c \in \Cbb^{M}$. Here, 
$\bm y^*$ is the optimal solution and $\bm \alpha^\ast$ is defined via 
$\bm y^*=\sum_i \alpha_i^* \bm a_i$. For Problem \ref{pro:lowrank}, we need the extra assumption that $\Amat = \Amat_1\Amat_2$, $\Amat_1 \in \Cbb^{N\times R},\Amat_2\in \Cbb^{R\times M}$ are two rank $R$ matrices, 
$\Gamma'$ is obtained by changing $\bm a_i$ into the $i$-th columns of $\Amat_1$ and $\Amat_2^{\dagger}$ in the definition of $\Gamma$.
For all problems considered in the table, we assume $\vabs{\Amat} = \Tht{1}, \vabs{\bm b} = \Ord{1},\vabs{\bm c}=\Ord{1}, \vabs{\Amat_1}= \Ord{1}$, $\vabs{\Amat_2} = \Ord{1}$, and $\varepsilon = \Ord{1}$.} 
\label{tab:ComparisonAlg}
\end{table*}

\section*{Acknowledgement}
We would like to thank Yassine Hamoudi, Miklos Santha and Armando Bellante for helpful discussions.
M.R., L.Z., and P.R. acknowledge support from Singapore’s Ministry of Education and National Research Foundation.
B.W. would like to thank CQT (NUS) for supporting a visit where this work began.

\begin{appendix}

\section{Comparison of the input model\label{app:CompareInput}}

Consider the matrix 
\begin{equation}
    \tilde{\Amat} = \sum_{i\leq K} \alpha_{i}\tilde{\Umat_{i}},
    \label{eq:ELCUrep}
\end{equation}
where $ \tilde{\Amat} \in \Cbb^{\tilde{N}\times \tilde{N}}, \alpha_{i}\in \Rbb$, $ K = \polylog \tilde{N}$, and each $\tilde{\Umat}_{j}$ has depth $\polylog \tilde{N}$. 
We say $\Amat$ has an \emph{efficient linear combination of unitaries} ({ELCU}) representation if it has the above representation.
In the following, we assume to be given an $N\times M$ matrix $\Amat$, with $M =2^m$ and $N=2^n$ and where each column is normalized and can be accessed by a quantum circuit efficiently. The cost function of the linear system $\argmin_{\bm x} \vabs{\Amat \bm x - \bm b}$ can be converted into a new cost function $\argmin_{\bm x} \vabs{\bra{0^{ m}}\tilde{\Amat}\ket{0^{ m}}\bm x - \frac{2}{M}\bm b}$, where $\tilde{\Amat} \in \Cbb^{\tilde{N}\times \tilde{N}}$ and $\tilde{N}=MN$, has an ELCU representation.

Let $\Umat_i$ be the unitaries preparing the columns of $\Amat\in \Cbb^{N\times M}$, where $i\in[M]$. From the outset the column norms of $A$ shall be $1$ for simplicity.
Given $\Umat_{\rm comb} = \sum_{i=1}^M \ket i \bra i \otimes \Umat_i$.
Also consider the copy operation for basis vectors $\Umat_{\rm copy} \ket i \ket {0^n} = \ket i \ket i$. Note that the second register can be larger than the first.
The following sequence of unitaries generates a block encoding of $\Amat$:
\begin{widetext}
\be
\tilde{\Umat}_{A1} &:=& (H^{\otimes m} \otimes \Imat) \Umat_{\rm comb} (\Imat- 2 \sum_{i=1}^M \ket i \bra i \otimes \ket {0^n} \bra {0^n}) \Umat_{\rm copy}^\dagger (H^{\otimes  m} \otimes \Imat) \\
&=& 
(H^{\otimes  m} \otimes \Imat)( \Umat_{\rm comb} \Umat_{\rm copy}^\dagger - 2 \sum_{i=1}^M \ket i \bra i \otimes \Umat_i \ket{0^n} \bra i)(H^{\otimes  m} \otimes \Imat) \\
&=&
(H^{\otimes  m} \otimes \Imat)( \Umat_{\rm comb} \Umat_{\rm copy}^\dagger - 2 \sum_{i=1}^M \ket i \bra i \otimes \ket{ a_i} \bra i)(H^{\otimes  m} \otimes \Imat) \\
&=&
(H^{\otimes  m} \otimes \Imat) \Umat_{\rm comb} \Umat_{\rm copy}^\dagger(H^{\otimes  m} \otimes \Imat) - \frac{2}{M} \ket 0^m  \bra 0^m \sum_{i=1}^M  \ket{ a_i} \bra i + P_\perp.
\ee
\end{widetext}

Here, $P_\perp $ is an operator that is zero in the $\ket {0^m} \bra {0^m}$ block. 
Note that the first term is the unitary 
\be
\tilde{\Umat}_{A2} := (H^{\otimes m} \otimes \Imat) \Umat_{\rm comb} \Umat_{\rm copy}^\dagger(H^{\otimes m} \otimes \Imat).
\ee 
Define the linear system with the matrix
\be
\tilde \Amat &=& \tilde{\Umat}_{A2} - \tilde{\Umat}_{A1}.
\ee
Naively we minimize $\Vert \tilde \Amat \ket 0  \bm x - \tilde {\bm b}\Vert $ using some $\tilde {\bm b}$ in the enhanced space. However, we can just minimize a different cost function
$\Vert \bra {0^m} \tilde \Amat \ket {0^m}   \bm x -  \frac{2}{M}\bm b \Vert $, where we use the original $\bm b$.
Noting that if $ {\bm x^\ast}$ is the solution to $A  {\bm x^\ast} = \bm b$, then 
\be
\bra {0^m}  \tilde \Amat \ket {0^m}   {\bm x^\ast} =  \bra {0^m}  \left ( \begin{array}{c}
    \frac{2}{M} \Amat  {\bm x^\ast}  \\
     P_\perp \ket {0^m}   {\bm x^\ast}
\end{array}\right) =  
    \frac{2}{M} \bm b 
.
\ee
We have used that $\bra 0 P_\perp \ket 0 = 0$. We have shown that our data input model can be formally translated into a small linear combination of unitaries. However the ability  to perform $\Umat_{\rm comb}$ can be considered to go beyond the NISQ setting. Note that Ref.~\cite{bravo2019variational} performs an analogous association of sparse matrix oracle access with the ECLU representation. 

Also note that the ELCU model can be transformed into the column model. However this comes at the price of many columns, which will not be solved efficiently by our algorithms. Future work may study sampling algorithms for this setting \cite{strohmer2009randomized}. 
The following shows that for matrix $\tilde{\Amat}$ given by Equation \eqref{eq:ELCUrep}, each column can be accessed with $\polylog N$-size quantum circuits (but we have $N$ of them).
 Let the $j$-th column of $\tilde \Amat$ be $\tilde \Amat_j$. Then
 \be
  \tilde \Amat_j &=& \pbra{\sum_{i = 1}^K \alpha_i \tilde \Umat_i} \ket{j} = \sum_{i = 1}^K \alpha_i \pbra{\tilde \Umat_i \ket{j}} \\
  &=&  \sum_{i = 1}^K \alpha_i\pbra{ \tilde \Umat_{i}^{(j)} \ket{0^n}},
\ee
where $\tilde \Umat_i^{(j)}\ket{0^n} = \tilde \Umat_i \ket{j} = \tilde \Umat_i C_j \ket{0^n}$, and $C_j$ is one layer of Pauli-$X$ gates to map $\ket{0^n}$ to $\ket{j}$.
Let 
\be\Umat_j := \sum_{i=1}^K \ket{i}\bra{i} \otimes  \tilde \Umat_i^{(j)},\ee
and similar to Ref.~\cite{berry2015simulating}, let
\begin{align}
  \Umat_{\rm comb}^{(j)} &:= \pbra{\Umat_{\rm prep}^\dagger (\alpha)\otimes \Imat} \Umat_j \pbra{\Umat_{\rm prep}(\alpha)\otimes \Imat}\\
&= \ket{0^{ \log K }}\bra{0^{ \log K }} \sum_{i=1}^{K} \alpha_i \tilde\Umat_i^{(j)} + \ket{\perp}\bra{\perp}.
\end{align}
Here, we have used 
\be
\Umat_{\rm prep}(\alpha) \ket {0^n} = \sum_{i=1}^K \sqrt{\alpha_i} \ket i,
\ee
which in general takes depth $\tOrd{K}$ to implement.
Then $\tilde \Amat_j = \bra{0^{ \log K }} \Umat_{\rm comb}^{(j)}\ket{0^{ \log K }} \ket{0^{n}}$.
Since there exist poly$(\log N)$-size quantum circuits for $\Umat_j$ and $\Umat_{\rm prep}(\alpha)$, there also exists a poly$(\log N)$-size quantum circuit for $\tilde \Amat_j$.
We emphasize again that the problem of this mapping is that $j \in [N]$ and $N$ is large in general.

\section{Algorithm for factorized linear system\label{app:FactorizedLS}}

Let $\Amat\in \Cbb^{N\times M}$ be a rank $R$ matrix. Let 
$\Amat_{1} :=\sum_{1\leq j\leq R} \vabs{\bm u_j}\ket{u_j}\bra{j}\in \Cbb^{N\times R}$ and 
$\Amat_{2} :=\sum_{1\leq j\leq R} \vabs{\bm v_j}\ket{j}\bra{v_j}\in \Cbb^{R\times M}$ be two rank $R$ matrices, such that $\Amat_{1} \Amat_{2} = \Amat$.
The matrices $\Amat_{1}$ and $\Amat_{2}$ are defined by the quantum states $\ket{u_j}$ and $\ket{v_j}$ and corresponding norms, respectively, analogous to Assumption \ref{assume:quantummatrix}.
In addition,  $\bm{b}\in \mathbb{C}^{N}$ and $\ket{b} := \bm b/\vabs{\bm b}$ is defined analogous to Assumption \ref{assume:vecb}. Let $\Gamma = \max_{1\leq j\leq R}\cbra{\vabs{\bm u_j}^2, \vabs{\bm u_j}\vabs{\bm b}, \vabs{\bm v_j}^2}$. The following theorem states that when the rank $R= \Ord{\polylog\pbra{M,N}}$, we can solve the linear system $\argmin_x \vabs{\Amat \bm x - \bm b}$ in a run time dependent on $\polylog \pbra{M,N}$. 
\begin{theorem}
For matrix $\Amat$ and vector $\bm b$ as mentioned above, there exists an $$
\Ord{R^{3}\Gamma^2\pbra{\kappa(\Amat_1)\kappa(\Amat_2)}^8\poly\pbra{\log M,\log N, \vabs{\bm b}}/\varepsilon^2}
$$
time hybrid algorithm, which outputs
a classical vector $\bm s = \pbra{s_1, \cdots, s_R}$, such that the approximate solution
$\bm \hat{\bm x} = \sum_{1\leq j\leq R} s_j \ket{v_j}$ satisfies $ \vabs{\Amat\hat{\bm x} - \bm b} - \min_x \vabs{\Amat \bm x - \bm b}\leq \varepsilon $ with high success probability.
\label{thm:AppLowRank}
 \end{theorem}

\begin{proof}
Let $\bm x^{\ast} := \argmin_{\bm x}\vabs{\Amat \bm x - \bm b}$, and $\bm y^{\ast} := \Amat_{2} \bm x^{\ast}$, then $\|\Amat \bm x^{\ast} - \bm b\| = \|\Amat_{1}\bm y^{\ast} - \bm b\|$. In the following we output $\bm s= (s_{1}, \cdots, s_{R})$ such that $\vabs{\Amat \sum_{1\leq j \leq R} s_{j}\ket{v_{j}}-\bm b} -\vabs{\Amat \bm x^{\ast} - \bm b}\leq \varepsilon$ by solving the following two linear system problems:
\begin{itemize}
    \item [(a)] Output an approximate solution $\hat{\bm y}$ such that
    $\vabs{\Amat_1 \hat{\bm y} - \bm b} - \vabs{\Amat_1 \bm y^{\ast} - \bm b}\leq \varepsilon_1$;
    \item[(b)] Solve linear system $\min_{\bm x}\vabs{\Amat_{2} \bm x - \bm y}$, and output the coefficients $\bm s = \pbra{s_{1}, \cdots, s_{R}}$ of the approximate solution $\hat{\bm x} = \sum_{1\leq j \leq R} s_j \ket{v_j}$ such that $\vabs{\hat{\bm x} - \bm x^{\ast}}\leq \varepsilon/\vabs{\Amat}$.
\end{itemize}
Here, $\varepsilon_{1}$ depends on Problem (b), and we will give its explicit value later. Problems (a) and (b) can be solved with our Algorithms \ref{Alg:overdetermined} and \ref{Alg:underdetermined}, respectively. Nevertheless, we can simplify this process since matrices $\Amat_1, \Amat_2$ are both full rank.

Let $\Gamma_{1}=\max\cbra{\vabs{\bm u_j}^2, \vabs{\bm u_j}\vabs{\bm b}}, \Gamma_2 = \max_{j}\vabs{\bm v_j}^2$.
The algorithm is as follows. First, generate approximate matrix $\hat{\Vmat}$ of matrix $\Vmat=\Amat_{1}^{\dagger}\Amat_1$, and approximate vector $\hat{\bm q}$ of vector $\bm q = \Amat_1 \bm b$ by Hadamard tests (similar to Algorithms \ref{Alg:overdetermined} and \ref{Alg:underdetermined}) with $\Gamma_{1}^{2}R^{2}T_{1}$ independent quantum measurements. In the same way, generate approximate materice $\hat{\Qmat}$ of matrix $\Qmat=\Amat_2\Amat_2^{\dagger}$ with $\Gamma_{2}R^{2} T_{2}$ independent quantum measurements.
 
Next, calculate $\hat{\bm y} = \hat{\Vmat}^{-1} \hat{\bm q}$ classically and output $\hat{\bm y}$ as a solution of Problem (a). Then by Proposition 10 of Huang, Bharti and Rebentrost~\cite{huang2019near}, we have $\vabs{\hat{\bm y} - \bm y^{\ast}}\leq \varepsilon_{1}$ when $T_{1} = \Ord{R \vabs{\Vmat^{-1}}^{2}\pbra{1 + \vabs{\bm y^{\ast}}^{2}}/\varepsilon_{1}^2}$.
Finally, calculate $\hat{\bm \alpha} = \hat{\Qmat}^{-1}\hat{\bm y}$ classically, and output the vector  $\pbra{\hat{\alpha}_{1}\vabs{\bm v_{1}}, \cdots, \hat{\alpha}_{R}\vabs{\bm v_{R}}}$ as a solution of Problem (b). 

To bound the error of the output, we require $\vabs{\hat{\bm y} - \bm y^{\ast}}\leq \Ord{\sqrt{R/T_{2}}}$, and 
\be T_{2} = \Ord{R \vabs{\Amat}^{2}\vabs{\Qmat^{-1}}^{2}\pbra{1 + \vabs{\bm\alpha^{*}}^{2}}/\varepsilon^2}\ee
by Huang, Bharti and Rebentrost~\cite{huang2019near}. Let $\varepsilon_{1} = \sqrt{R/T_{2}}$, which gives us 
$$ T_{1} = \frac{R\vabs{\Amat}^{2}\vabs{\Vmat^{-1}}^{2}\vabs{\Qmat^{-1}}^{2}\pbra{1 + \vabs{\bm y^{\ast}}^{2}}\pbra{1 + \vabs{\bm \alpha^{\ast}}^{2}}}{\varepsilon^2},$$
and thus the total number of quantum measurements is bounded by 
\begin{widetext}
\be
\Gamma^{2} R^{2}\pbra{T_{1} + T_{2}} &=& \Ord{\frac{\Gamma^{2} R^{3}\vabs{\Amat}^{2}\vabs{\Qmat^{-1}}^{2}\pbra{\vabs{\Vmat^{-1}}^{2}\pbra{1 + \vabs{\bm y^{\ast}}^{2}} + 1}\pbra{1 + \vabs{\bm \alpha^{\ast}}^{2}}}{\varepsilon^2}}\\
&=& \Ord{\frac{\Gamma^{2} R^{3}\vabs{\Amat}^{2}\vabs{\Amat_{2}^{-1}}^{4}\pbra{\vabs{\Amat_{1}^{-1}}^{4}\pbra{1 + \vabs{\bm y^{\ast}}^{2}} + 1}\pbra{1 + \vabs{\bm \alpha^{\ast}}^{2}}}{\varepsilon^2}},
\ee
\end{widetext}
where $\Gamma = \max\cbra{\Gamma_{1}, \Gamma_{2}}$, which can be simplified to 
$\Ord{\frac{\Gamma^2 R^3  \kappa\pbra{\Amat_1}^6 \kappa(\Amat_2)^6\kappa\pbra{\Amat}^2}{\varepsilon^2}}$
with the assumption that $\vabs{\bm b}=\Ord{1}, \vabs{\Amat_1} = \Ord{1}, \vabs{\Amat_2} = \Ord{1}$.
\end{proof}

\section{Algorithm for ``rank-relaxed" factorized linear system\label{app:GeneralFactorizedLS}}

Let $\Amat\in \Cbb^{N\times M}$, 
$\Amat_{1} :=\sum_{1\leq j\leq R} \vabs{\bm u_j}\ket{u_j}\bra{j}\in \Cbb^{N\times R}$, and 
$\Amat_{2} :=\sum_{1\leq j\leq R} \vabs{\bm v_j}\ket{j}\bra{v_j}\in \Cbb^{R\times M}$ with $\rank\pbra{\Amat_2} = R$, such that $\Amat_{1} \Amat_{2} = \Amat$. Let also $\bm{b}\in \mathbb{C}^{N}$. The vectors $\ket{u_j}, \ket{v_j}$ and $\ket{b} := \bm b/\vabs{\bm b}$ can be quantum accessed analogously to Assumptions \ref{assume:quantummatrix} and \ref{assume:vecb}. Let the size bound $\Gamma = \max_{1\leq j\leq R}\cbra{\vabs{\bm u_j}^2, \vabs{\bm u_j}\vabs{\bm b}, \vabs{\bm v_j}^2}$. The following theorem states that we can solve the linear system $\argmin_x \vabs{\Amat \bm x - \bm b}^2$ in a run time dependent on $\polylog \pbra{M,N}$.
This algorithm depends on $1/\varepsilon^4$  linearly, since we relax the decomposition form of $\Amat$, and $\Amat_1$ can be a low rank matrix.
\begin{theorem}
For matrix $\Amat$ and vector $\bm b$ as mentioned above, there exists an \be\Ord{R^{3}\Gamma^2\poly\pbra{\kappa\pbra{\Amat_1}, \kappa\pbra{\Amat_2},\log M,\log N, \vabs{\bm b}}/\varepsilon^4}\ee
time hybrid algorithm, which outputs
a classical vector $\bm s = \pbra{s_1, \cdots, s_R}$ such that the approximate solution
$\bm \hat{\bm x} = \sum_{1\leq j\leq R} s_j \ket{v_j}$, and $ \vabs{\Amat\hat{\bm x} - \bm b} - \min_x \vabs{\Amat \bm x - \bm b}\leq \varepsilon $ with high success probability.
\label{thm:AppGFLS}
 \end{theorem}

\begin{proof}
For matrix $\Amat\in \Cbb^{N\times M}$ and $\Amat = \Amat_1\Amat_2$, $\Amat_2\in \Cbb^{R\times M}$ is full row rank.
Let $\bm y := \Amat_{2} \bm x$, then $\Amat_{1}\bm y = \bm b$. Let $\Vmat_{1} = \Amat_{1}^{\dagger}\Amat_1$, and $\Gamma_{1}$ be the maximum absolute value of all of elements of $\Vmat_{1}$ and $\bm b$.
Similar to Algorithm \ref{Alg:overdetermined}, we can generate an approximate matrix $\hat{\Vmat}_{1}$ of $\Vmat_{1}$ by $\Gamma_{1}^{2}R^{2}T_{1}$ independent measurements of some quantum circuits such that $\vabs{\hat{\Vmat}_{1} - \Vmat_{1}}\leq \Ord{\sqrt{R/T_1}}$.

Let $\hat{\bm y}$ be the solution of the approximate and perturbed system $\hat{\Wmat}=\hat{\Vmat}_{1} + \lambda \Imat$, with $\lambda = \varepsilon_{1}/\pbra{2\vabs{\Amat_{1}^{-1}}^4\vabs{\Amat_{1}}^2\vabs{\bm b}}$.
Then we have $\vabs{\hat{\bm y} - \bm y}\leq \varepsilon_{1}$, with
$$ T_{1} = \frac{ cR \vabs{\Amat_{1}^{-1}}^8\vabs{\Amat_{1}}^2\vabs{\bm b}^2\pbra{\vabs{\Amat_1}(\vabs{\bm x^*} + 1) + \varepsilon_1}^2 }{\varepsilon_{1}^4},$$
for large constant $c$.

For the system $\Amat_{2}\bm x = \bm y$, let $\Vmat_{2} = \Amat_{2}\Amat_{2}^{\dagger}$. By Lemma \ref{lem:decompX}, there exists a vector $\bm \alpha\in \Cbb^{R}$ such that $\bm x = \sum_{i}\alpha_{i}\vabs{\bm a_{i}} \ket{a_{i}}$, and thus
$\Vmat_{2}\bm \alpha = \Amat_{2} \bm x$. Since $\Amat_{2}$ is full row rank, $\Vmat_{2}$ is full rank. Let $\hat{\Vmat}_{2}$ be the approximated matrix of $\Vmat_{2}$ with the same method as $\hat{\Vmat}$ in Algorithm \ref{Alg:underdetermined}. Then $\hat{\Vmat}_{2}$ can be obtained by $\Ord{\Gamma_2 R^{2} T_{2}}$ independent measurements of Hadamard tests, where $\Gamma_2$ is the maximum absolute value of all elements in $\Vmat_2$, and $\|\hat{\Vmat} - \Vmat\|\leq \Ord{\sqrt{R/T_{2}}}$.
Let $\hat{\bm \alpha}=\hat{\Vmat}^{-1}\bm b$ be the solution of the approximate system, and $\varepsilon_{1} = \sqrt{R/T_{2}}$. Since $\Vmat_{2}$ is invertible, by Proposition 10 of Huang, Bharti and Rebentrost~\cite{huang2019near}, $\|\Amat_{2}\hat{\bm x} - \bm y\| - \|\Amat_{2}\bm x - \bm y\| \leq \varepsilon/\vabs{\Amat_1}$, with $T_{2} = \Ord{R \vabs{\Amat_1}^2\vabs{\Vmat_{2}} \vabs{\Vmat_{2}^{-1}}^{2}(1+\vabs{\bm\alpha^\ast})^{2}/\varepsilon^2} = \Ord{R \vabs{\Amat_1}^2\vabs{\Amat_{2}}^{2} \vabs{\Amat_{2}^{-1}}^{4}(1+\vabs{\bm \alpha^\ast})^{2}/\varepsilon}$, where $\bm \alpha^\ast =\Vmat_2^{-1}\bm y$. Hence, 
$\vabs{\Amat \bm \hat{\bm x} - \bm y} - \vabs{\Amat \bm x - \bm y}\leq \varepsilon$, with quantum measurements
\be
&&\Ord{\Gamma^{2}R^{2}\pbra{T_{1} + T_{2}}}\\
&= &\Gamma^{2}R^{3}\Ord{ \frac{a_1 a_2^2 (\vabs{\bm x^\ast} + 1)^2}{\varepsilon^4} + \frac{(a_1+1) a_2}{\varepsilon^2}}
\ee
where 
$a_1 = \vabs{\Amat_{1}^{-1}}^8\vabs{\Amat_{1}}^4\vabs{\bm b}^2, a_2 = \vabs{\Amat_1}^2\vabs{\Amat_{2}}^2 \vabs{\Amat_{2}^{-1}}^{4}(1+\vabs{\bm\alpha^\ast})^{2}$, and
$\Gamma = \max (\Gamma_1, \Gamma_2)$. 
\end{proof}

\section{\label{app:ApplicationErrorBound}Application of the hybrid output in Algorithm \ref{Alg:underdetermined}}
In this section, we will give the error analysis for the two applications of our hybrid outputs, as stated in Lemmas \ref{clm:ApplicationError} and \ref{clm:ApplicationError2}.

Let $\hat{\bm y} = \sum_{i = 1}^{M} \hat{\alpha_i} \vabs{\bm a_i}\ket{a_i}$ be an approximation of optimal solution $\bm y = \sum_{i = 1}^M \alpha_i \vabs{\bm a_i}\ket{a_i}$ such that $\vabs{\hat{\bm \alpha} - \bm \alpha}\leq \eta$, where $\eta>0$ is the error of our algorithm. The more measurements we have in our Algorithm \ref{Alg:underdetermined}, the smaller 
the error $\eta$.
The first application is measuring an overlap with an arbitrary quantum state. 
\begin{lemma}
Given $\bm y$ and $\hat{\bm y}$ as defined above, $\varepsilon\in (0, 1)$, and a poly-logarithmic depth circuit $\Umat_v \in \Cbb^{N\times N}$.
Let $\ket{v} = \Umat_v\ket{0^{\log N}}$, and $\eta = \frac{\varepsilon}{2\vabs{\Amat}_F}$. 
There exists an algorithm which outputs the inner product $\widetilde{\bra{v}\bm y}$, such that $\abs{\widetilde{\bra{v}{\bm y}} - \bra{v}\bm y}\leq \varepsilon$ with $\Ord{\frac{M\vabs{\Amat}_F^2\vabs{\bm \alpha}^2}{\varepsilon^2}}$ independent measurements.
\label{clm:ApplicationError}
\end{lemma}

\begin{proof}
We can use Hadamard test to approximate $\braket{v|a_i} = \bra{0^{\log N}}\Umat_v^{\dagger} \Umat_i\ket{0^{\log N}}$. Suppose we independently take $T$ measurements of the Hadamard test $\cbra{\ket{0^{\log N}}, \Umat_v^{\dagger} \Umat_i}$ and obtain the approximation  $\widetilde{\braket{v|a_i}}$. Then we have 
 $ \abs{\widetilde{\braket{v|a_i}} - \braket{v|a_i}}\leq \Ord{1/\sqrt{T}}$.
 Let $\widetilde{\bra{v}\bm y} = \sum_i \hat{\alpha}_i \vabs{\bm a_i}\widetilde{\braket{v|a_i}}$.
 Hence
 \begin{widetext}
\begin{align}
       \abs{\widetilde{\bra{v}{\bm y}} - \bra{v}\bm y} &\leq \abs{\widetilde{\bra{v}\bm y} - \bra{v}\hat{\bm y}} +   \abs{\bra{v}\hat{\bm y} - \bra{v}\bm y}\\
   & \leq \abs{\sum_{i = 1}^M \hat\alpha_i \vabs{\bm a_i}\Ord{\sqrt{1/T}}} + \abs{ \sum_{i = 1}^M \hat{\alpha}_i\vabs{\bm a_i}\braket{v|a_i} - \sum_{i = 1}^M \alpha_i \vabs{\bm a_i} \braket{v|a_i}} \\
   &\leq \vabs{\hat{\bm \alpha}}\vabs{\Amat}_F \Ord{\sqrt{1/T}} + \vabs{\bm \alpha - \hat{\bm \alpha}}\vabs{\Amat}_F \\
   &\leq \vabs{\bm \alpha}\vabs{\Amat}_F \Ord{\sqrt{1/T}} +
   \vabs{\hat{\bm \alpha} - \bm \alpha}\vabs{\Amat}_F \Ord{\sqrt{1/T}}
  + \vabs{\bm \alpha - \hat{\bm \alpha}}\vabs{\Amat}_F\\
   &\leq \varepsilon
\end{align}
 \end{widetext}

The last inequality holds when $T = \Ord{\pbra{\varepsilon + \vabs{\Amat}_F\vabs{\bm \alpha}}^2/\varepsilon^2}$. The expression can be simplified to $T = \Ord{\vabs{\Amat}_F^2\vabs{\bm \alpha}^2/\varepsilon^2}$ with the fact that $T=\Omega(1)$. Since we need to approximate all of $\braket{v|a_i}$ for $i\in[M]$, the total measurements equals $\Ord{M\vabs{\Amat}_F^2\vabs{\bm \alpha}^2/\varepsilon^2}$.
\end{proof}
The second application is measuring an observable on the output. 
\begin{lemma}
Given $\bm y, \hat{\bm y} $ as defined above, $K_H = \Ord{\polylog N}$,  $\varepsilon\leq \min\cbra{1, \Delta_H K_H \vabs{\Amat}_F^2 \vabs{\bm \alpha}^2}$, and poly-logarithmic depth quantum circuits $\Hmat_i$ for $i\in K_H$.
Let $\Hmat = \sum_{i\leq K_H} \gamma_i \Hmat_i$, where $0<\gamma_i<\Delta_H$ for $i\in[K_H]$ are coefficients.
$\eta =\min\cbra{ \frac{\varepsilon}{4\vabs{\Amat}_F^2 \vabs{\bm \alpha}\Delta_H K_H}, \vabs{\bm \alpha}\varepsilon}$.
There exists an algorithm which outputs $\widetilde{\bm y^{\dagger} \Hmat\bm y}$ such that $\abs{ \widetilde{\bm y^{\dagger} \Hmat\bm y} - \bm y^{\dagger} \Hmat \bm y}\leq \varepsilon$ with $\tOrd{\frac{M^2\Delta_H^2\vabs{\Amat}_F^4\vabs{\bm \alpha}^4}{\varepsilon^2}}$ independent measurements.
\label{clm:ApplicationError2}
\end{lemma}

\begin{proof}
Let $\widetilde{\bra{a_i}\Hmat_j\ket{a_k }}$ denote the approximation of $\bra{a_i} \Hmat_j\ket{a_k }$ by $T$ independent measurements of the Hadamard test $\cbra{\ket{0^{\log N}}, \Umat_i^{\dagger} \Hmat_j \Umat_k}$. Note that
$\widetilde{\bm y^{\dagger} \Hmat \bm y} = \sum_{i,j,k} \hat{\alpha}_i^{\ast} \gamma_j \hat{\alpha}_k \vabs{\bm a_i}\vabs{\bm a_k}\widetilde{\bra{a_i}\Hmat_j\ket{a_k}}$.
Then we have
\begin{widetext}
\begin{align}
     \abs{\widetilde{\bm y^{\dagger}\Hmat \bm y} - \bm y^{\dagger}\Hmat \bm y}
     &\leq
    \abs{\widetilde{\bm y^{\dagger}\Hmat \bm y} - \hat{\bm y}^{\dagger}\Hmat \hat{\bm y}} + \abs{\hat{\bm y}^{\dagger}\Hmat \hat{\bm y} - \bm y^{\dagger}\Hmat \bm y}\\
    &\leq \abs{\sum_{i,j,k} \hat{\alpha}_i^{\ast} \gamma_j \hat{\alpha}_k \vabs{\bm a_i}\vabs{\bm a_k}
   \Ord{ \sqrt{1/T}}} \\
   &+ \abs{\sum_{j} \gamma_j \sum_{i,k}\pbra{\hat{\alpha}_i^{\ast}\hat{\alpha}_k -\alpha_i^\ast \alpha_k}
    \vabs{\bm a_i}\vabs{\bm a_k}\bra{a_i }\Hmat_j \ket{a_k}}\\
    &\leq \Delta_H K_H\vabs{\Amat}_F^2\Ord{\sqrt{1/T}} \pbra{\eta^2 + \vabs{\bm \alpha}^2} \\
    &+ 
  \Delta_H K_H   \vabs{\Amat}_F^2 \sqrt{\sum_{i,k} \pbra{\hat{\alpha}_i^{\ast}\hat{\alpha}_k - \alpha_i^{\ast} \alpha_k}^2}\\
    &\leq \varepsilon.
\end{align}
\end{widetext}
The last inequality holds when 
(i) $\sum_{i,k} \pbra{\hat{\alpha}_i^{\ast}\hat{\alpha}_k - \alpha_i^{\ast} \alpha_k}^2 = \Ord{\vabs{\bm \alpha}^2 \eta^2}$ and (ii) $T = \frac{c\Delta_H^2 K_H^2 \vabs{\Amat}_F^4 \pbra{\eta^2 + \vabs{\bm \alpha}^2}^2}{\varepsilon^2}$ for a suitable constant $c$. 
The number of measurements can be simplified to 
\begin{align}
    T &= \frac{c'\Delta_H^2 K_H^2 \vabs{\Amat}_F^4 \vabs{\bm \alpha}^4}{\varepsilon^2} + \frac{c'\varepsilon^2}{\Delta_H^2 K_H^2 \vabs{\Amat}_F^4\vabs{\bm \alpha}^4} \\
&= \Ord{ \frac{\Delta_H^2 K_H^2 \vabs{\Amat}_F^4 \vabs{\bm \alpha}^4}{\varepsilon^2}},
\end{align}
for a suitable constant $c'$, since $\varepsilon\leq \Delta_H K_H \vabs{\Amat}_F^2 \vabs{\bm \alpha}^2$.
To show (i),  let $ \abs{\alpha_i - \hat\alpha_i} = \varepsilon_i$. Then $\sum_i \varepsilon_i^2 = \eta^2$.
 By the definition of $\bm \alpha$, we have $\alpha_i>0$. 
Hence,
\begin{align}
\begin{aligned}
    \sum_{i,k} \pbra{\hat{\alpha}_i^{\ast}\hat{\alpha}_k - \alpha_i^{\ast} \alpha_k}^2 &\leq \sum_{i,k}\pbra{\pbra{\alpha_i + \varepsilon_i} \pbra{\alpha_k + \varepsilon_k} - \alpha_i \alpha_k}^2\\
    &\leq 4\vabs{\bm \alpha}^2 \eta^2 + 4\vabs{\bm \alpha} \eta^3 + \eta^4\\
    &\leq 9\vabs{\bm \alpha}^2\eta^2.
\end{aligned}
\end{align}
The last inequality holds since $\eta\leq \vabs{\bm \alpha}$.
\end{proof}

\section{Optimized CNOT circuit for copy operation\label{app:copy}}

\begin{proof}
By Refs. \cite{moore2001parallel, jiang2020optimal}, the CNOT circuit constructed by $\CNOT_{1,j}$ for $2\leq j \leq n$ can be represented by a matrix $\Mmat\in \cbra{0,1}^{n\times n}$, in which only the first column and the diagonal elements are $1$, \emph{i.e.},
\begin{align}
    \Mmat = \begin{pmatrix}
    1 & 0 & 0 &\cdots & 0\\
    1 & 1 & 0 &\cdots & 0\\
    1 & 0 & 1 & \cdots & 0\\
    &&\cdots &&\\
    1 & 0 & 0 & \cdots & 1
    \end{pmatrix}.
    \label{eq:MatrixCnotOneCTGate}
\end{align}
The representation is unique, \emph{i.e.}, the CNOT circuits with the same matrix representation $\Mmat$ are equivalent.
A CNOT circuit with matrix representation $\Mmat$ followed by a CNOT gate $\CNOT_{i,j}$ is equivalent to add row $i$ to row $j$ of $\Mmat$.
Thus the construction of a CNOT circuit is equivalent to going from $\Imat$ to $\Mmat$ with paralleled row additions.
Here, row addition ${\rm AddRow}(\Mmat,i,j)$ means add row $i$ to row $j$ of matrix $\Mmat$.
Paralleled row addition means to do row additions in parallel.
The paralleled row addition algorithm from $\Imat$ to $\Mmat$ is as follows:
 \begin{itemize}
     \item Perform ${\rm AddRow} (\Mmat,2^{k-1}(2j + 1), 2^{k}j)$ for $j$ from $0$ to $\frac{n}{2^k}-1$, where $ k $ increases from $1$ to $\log n - 1$, except for all the row operations about the first row, \emph{i.e.}, without performing ${\rm AddRow} (\Mmat, j,0)$;
     \item Perform ${\rm AddRow} (\Mmat, 2^{k-1}(2j + 1), 2^{k}j)$ for $j$ from $0$ to $\frac{n}{2^k}-1$, and $ k $ decrease from $\log n$ to 1.
 \end{itemize}
For convenience here we suppose $\log n$ is an integer.  This assumption is easy to be generalized to the non-integer case.  Notice that CNOT$_{ji}$ is equivalent to add row $j$ to row $i$, and the transpose of $\Mmat$ is corresponding to transformation: $(x_1, x_2, \dots, x_n) \xrightarrow{\Mmat^T} \bm (x_1+\cdots + x_n, x_2, \cdots, x_n)$.  Hence, a simple way to verify the above process is to check whether the transpose operation of the above process equals $\Mmat^T$.
Since for each fixed $k$, all of row operations can be performed in parallel, there are $2\ceil{\log n}-1$ paralleled row additions in total.
\end{proof}

\section{
Optimization of a control circuit
\label{app:onelayercontrol}}
In this section, we give the proof of Lemma \ref{lem:OneLayerLog}.

\begin{proof}[Proof of Lemma \ref{lem:OneLayerLog}]
Observe that any one layer of a quantum circuit can be represented as
\be \Ccal = \bigotimes_{k\leq s}\CNOT_{i_{k},j_{k}}\bigotimes_{p\leq t} \Rmat_{p}(\theta),\ee 
where $i_{k},j_{k},p$ are disjoint qubits for any $k\in[s], p\in[t]$, and $s+t\leq n$.
In the following, we optimize the depth of Control-CNOTs and  Control-$\Rmat(\theta)$ respectively.

The controlled CNOT circuits which contain $s$ CNOTs are actually $s$ Toffoli gates which share one common control qubit.
Observe that we can greatly reduce the depth by using additional qubits to avoid the concentration of the control qubit of the Toffoli gates to the single ancilla by the idea of Ref.~\cite{barenco1995elementary}, as depicted in Fig. \ref{cir:ToffoliEquivalent}.
We say an $n+m$-qubit unitary $\Umat$ implements an $n$-qubit matrix $\Mmat$ with $m$ borrowed ancilla qubits if $\Umat(\ket{x}\ket{a}) =(\Mmat\ket{x})\ket{a}$ for any $a\in \cbra{0,1}^m$ and $x\in\cbra{0,1}^n$.
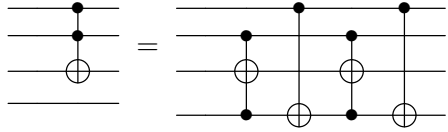
\begin{figure}[ht]
\centering\[
\Qcircuit @C=1em @R=.7em {
&\qw & \ctrl{1} & \qw\\
&\qw & \ctrl{1} & \qw\\
&\qw & \targ & \qw\\
&\qw & \qw & \qw
}
\Qcircuit @C=1em @R=.7em {
&\\
&\\
&=\\
&
}\quad
\Qcircuit @C=1em @R=.7em {
&\qw &\qw & \ctrl{3} & \qw & \ctrl{3}& \qw\\
&\qw & \ctrl{1} &\qw & \ctrl{1} & \qw & \qw\\
&\qw & \targ & \qw & \targ &\qw & \qw\\
&\qw & \ctrl{-1}& \targ & \ctrl{-1} & \targ & \qw
}
\]
\caption{\cite{barenco1995elementary}Implement Toffoli gate with one borrowed ancilla qubit.}\label{cir:ToffoliEquivalent}
\end{figure}

The main idea of our optimization algorithm is to divide the CNOT gates of $\Ccal$ into two sets $S_1$ and $ S_2$, in which $0\leq 
s_1 - s_2\leq 1$, and $s_1,s_2$ are the size of $S_1,S_2$ respectively. 
First optimize the depth of the Control-CNOTs in set $S_1$ by utilizing the rests qubits (qubits in set $S_2$) as the borrowed ancillas.
And then similarly optimize the depth of Control-CNOTs in set $S_2$.

In details, 
we randomly choose $s_1$ qubits corresponding to gates of $S_2$ as borrowed ancillas, and replace all of Control-CNOTs for CNOTs in set $S_1$ to a new equivalent circuit, 
as in the right hand side of Fig. \ref{cir:ToffoliEquivalent}, 
by using these $s_{1}$ qubits as borrowed ancillas.
Without loss of generality, let $S_1:=\cbra{\CNOT_{1,2}, \cdots, \CNOT_{2s_1-1, 2s_1}}$, and $S_2:=\cbra{\CNOT_{2s_1+1,2s_1+2}, \cdots, \CNOT_{2s_1+2s_2-1, 2s_1+2s_2}}$.
This process can be implemented by performing:
\begin{itemize}
    \item [(1)] Toffoli$(2i, 2s_1 + i, 2i+1)$
    for $1\leq i\leq s_1$;
    \item [(2)] $\CNOT_{1,2s_1+i}$ for $1\leq i \leq s_1$;
    \item [(3)] Toffoli$(2i, 2s_1 + i, 2i+1)$ for $1\leq i\leq s_1$;
     \item [(4)] $\CNOT_{1,2s_1+i}$ for $1\leq i \leq s_1$;
\end{itemize}
By Claim \ref{claim:CnotOneCTGate}, both step (2) and (4) can be paralleled to $2\ceil{\log s_{1}} - 1$-depth respectively. The depth of Toffoli gate can be reduced to $8$ by Amy \cite{amy2013algorithms}, thus the total depth of step (1-4) are less than $ 4\ceil{\log s_{1}}+6$. Therefore, the depth of Control-CNOTs to construct Control-$\Ccal$ equals $4\ceil{\log s_{1}} + 4\ceil{\log s_{2}}+ 12\leq 8\ceil{\log s } +8$, the equality holds when $s_{1} = \ceil{s/2}$ and $s_{2} = \lfloor s/2\rfloor$.

\begin{figure}
\[
\Qcircuit @C=1em @R=.7em {
&\ctrl{1} & \qw \\
&\gate{\Rmat(\theta)} & \qw
}
\quad \Qcircuit @C=1em @R=.7em {
& &\\
&= &
}\quad
\Qcircuit @C=1em @R=.7em {
&\qw & \ctrl{1} &\qw &\ctrl{1} &\qw & \gate{e^{i\phi}} & \qw\\
& \gate{\Rmat(c)} & \targ & \gate{\Rmat(b)} &\targ &\gate{\Rmat(a)} &\qw &\qw
}
\]
    \caption{\cite{barenco1995elementary} Convert Control-$\Rmat(\theta)$ to CNOTs $+$ one-qubit gates, where $\Rmat(\theta) = e^{i\phi}\Rmat(a)X\Rmat(b)X\Rmat(c)$, and $\Rmat(a)\Rmat(b)\Rmat(c)=\Imat$, where $X$ is Pauli-$X$.}
    \label{cir:CUDecomp}
\end{figure}

For control of $t$ single-qubit gates, without loss of generality, let them be 
$$\Rmat_1(\theta_1),\Rmat_2(\theta_2),\cdots,\Rmat_t(\theta_t),$$ 
where $\Rmat_j(\theta)$ represents $\Rmat(\theta)$ operating on the $j$-th qubit.
 By Nielsen \emph{et al.} \cite{nielsen2002quantum}, any Control-$\Rmat(\theta)$ gate can be decomposed into four single-qubit gates combined with two CNOTs, as depicted in Fig. \ref{cir:CUDecomp}.
The circuit construction for the control of $t$ single-qubit gates is as follows: 
\begin{itemize}
    \item[(1)] Perform $\Rmat_j(c_j)$ for $1\leq j\leq t$;
    \item [(2)] Perform $\CNOT_{0,j}$ for $1\leq j\leq t$;
    \item[(3)] Perform $\Rmat_j(b_j)$ for $1\leq j\leq t$;
        \item [(4)] Perform $\CNOT_{0,j}$ for $1\leq j\leq t$;
    \item [(5)] Perform $\Rmat_j(c_j)$ for $1\leq j\leq t$;
    \item[(6)] Perform $e^{i\sum_j \phi_j}$ on the qubit 0;
\end{itemize}
where $0$ is the original control qubit. Since by Claim \ref{claim:CnotOneCTGate}, steps (2) and (4) can be optimized to $2\ceil{\log t} - 1$-depth respectively, and the depth of other steps all equal one, the depth of Control-$\cbra{\Rmat(\theta)}$ equals $4\ceil{\log t} + 1$.

Therefore, the total depth of Control-$\Ccal$ is less than $8\ceil{\log {s}} + 8 + 4\ceil{\log t} + 1\leq 12\ceil{\log n} +9$.
\end{proof}

\section{Optimize the depth of a CNOT circuit under lattice\label{app:CtrlCNOTLattice}}
In this section, we give the proof of Lemma \ref{lem:CNOTMap2lattice}. Before the proof, we first introduce a fact that there exists a one-to-one correspondence between an $n$-qubit CNOT circuit and an $n\times n$ invertible matrix. Specifically, the $j$-th column of the matrix for $n$-qubit CNOT circuit $\Ccal$ equals the $n$-boolean tuple obtained by operating $\Ccal$ on $\bm e_j = (0,\cdots,0, 1, 0, \cdots, 0)$ where only the $j$-th entry is 1 and all the other $n-1$ entries are 0. The construction holds because of the linearity of CNOT circuits, \emph{i.e.}, for any $i,j\in[n]$, $\Ccal\ket{i \oplus j} = \Ccal\ket{i} \oplus \Ccal\ket{j}$.  For example, a CNOT gate $\CNOT_{1,2}$ can be represented as matrix $\begin{pmatrix}
1 & 0 \\
1 & 1
\end{pmatrix}$\footnote{$\cdot \oplus\cdot $ is bit to bit XOR of two numbers.}. 

\begin{proof}[Proof of Lemma \ref{lem:CNOTMap2lattice}.]

Let CNOT circuit 
\be\Ccal_1 := \prod_{j=1}^k\CNOT_{1,i_j},\ee
where $i_j\ne i_{j'}$ for $j\ne j'$ and $i_j\in[n]$. Let set $S:=\cbra{1,i_1,...i_k}$. Let the top left corner vertex be $(1,1)$ in the $l_1\times l_2$ lattice. For vertex $(i,j)$, which is in row $i$ and column $j$ of the lattice.
First generate a spanning tree by removing all of the vertical edges connecting vertices $(i,j)$ and $(i+1,j)$ for $i\in [n-1] $ and $ 2\leq j\leq n$. Then recurrently
delete all of the leaves not in $S$, and the edges corresponding to these leaves, until all of the leaves are vertices in set $S$.
Denote this tree as tree $T$.
Let the root of tree $T$ be qubit on vertex $(1,1)$ of the lattice, denoted as qubit $1$ (we can use SWAP gates to change the root be qubit $1$), and use $c(i)$ to represent the set of children of qubit $i$.
The following algorithm gives an $\Ord{l_1+l_2}$-depth equivalent CNOT circuit under tree $T$ for $\Ccal_1$.
\begin{itemize}
    \item[(1)] 
  Let $\Tcal$ be the tree set which contains all of the maximal sub-trees $T_j$ of $T$ such that the root and leaves of $T_j$ are in $S$ and all of the remaining internal vertices are not in $S$.
  Permute the elements of $\Tcal$ such that if the tree $T_i$ is in front of the tree $T_j$, then the root of $T_j$ is not the descendant of the root of $T_i$ in tree $T$.
 Sequentially 
  take the maximum number of trees in $\Tcal$ until $\Tcal$ is empty,
   with the above order such that these selected trees have no intersection of their vertices. Let $d(T)$ denote the depth of tree $T$. Let the layer of root of $T$ be layer 1.
   For the selected trees $T_i$ in parallel:
  \begin{itemize}
      \item[(a)] 
      Perform $\CNOT_{i,c(i)}$ for $i$ in layer $d(T_i)-1$ decrease to layer $2$ sequentially.
      In one layer perform the CNOT operations in parallel if there is no collision.
      If $i$ has more than one child vertex, perform CNOT gates with control qubit being $i$ and target qubits being one of $c(i)$ sequentially.
      \item[(b)] 
      Perform $\CNOT_{i,c(i)}$ for $i$ in layer $1$ increase to layer $d(T_i)-1$ sequentially.  In one layer perform the CNOT operations in parallel if there is no collision. 
      If $i$ has more than one child vertex, perform CNOT gates with control qubit being $i$ and target qubits being one of $c(i)$ sequentially.
      \item[(c)] Repeatedly perform Steps (1)-(a),(1)-(b) except all of CNOT gates related to the leaves.
  \end{itemize}
    \item[(2)] Reverse the order of elements in $\Tcal$ to obtain a opposite sequence of tree as Step (1).
    Delete the first tree in the tree sequence which contains vertex 1 to get tree set $\Tcal'$.
    Sequentially take maximum number of trees in $\Tcal'$ with this order until $\Tcal$ is empty, such that these selected trees have no intersection of vertices. For the selected trees perform (1)(a-c) in parallel.
\end{itemize}
With this construction, we have the depth of $\Ccal_1$ is $\Ord{d(T)} = \Ord{l_1 + l_2}$.

For any leaves $i$ in $T_j$ with input $x_i$, it outputs  $x_i+x_k$ where $k$ is the root of $T_j$ after performing Step (1).
And then by Step (2), the input of vertex 1 (root of tree $T$) is added to all of the other vertices in $S$, letting all of the other vertices unchanged, and thus implementing $\Ccal_1$.
\end{proof}

\end{appendix}

\bibliographystyle{apsrev4-1}
 \bibliography{ref}

\end{document}